\documentclass[12pt,article]{memoir}
\usepackage[english.US]{babel}
\usepackage{amsthm}
\usepackage{amssymb}
\usepackage[scr]{rsfso}
\usepackage{upgreek}
\usepackage[leqno]{mathtools}
\usepackage{enumitem}
\usepackage{tikz}

\definecolor{refkey}{rgb}{0,0,1}
\definecolor{labelkey}{rgb}{1,0,0}

\usepackage{xr-hyper}
\usepackage{hyperref}

\usepackage{url}

\usepackage{todonotes}

\numberwithin{equation}{chapter}

\setsecheadstyle{\large\bfseries\raggedright}
\setsubsecheadstyle{\normalsize\bfseries\raggedright}

\theoremstyle{plain}
\newtheorem{theorem}{Theorem}
\newtheorem{proposition}[theorem]{Proposition}
\newtheorem{lemma}[theorem]{Lemma}
\newtheorem{corollary}[theorem]{Corollary}

\theoremstyle{definition}

\newtheorem{remark}[theorem]{Remark}

\newtheorem*{note}{Note}

\theoremstyle{plain}
\numberwithin{theorem}{chapter}

\theoremstyle{definition}
\newtoks{\thehRemark}
\newtheorem*{Remark}{\the\thehRemark}

\newenvironment{claim}[1][{\textup{(\theequation)}}]{\refstepcounter{equation}\vglue10pt
\begin{trivlist}
\item[{\hskip\labelsep#1}]}{\vglue10pt\end{trivlist}}

\newenvironment{phantomequation}[1][]{\refstepcounter{equation}}{}

\newcommand{\N}{\mathsf{N}}

\newcommand{\W}{\mathsf{W}}
\newcommand{\Dirac}{\mathsf{Dirac}}
\newcommand{\dist}{\mathsf{dist}}

\newcommand{\3}{{|\!|\!|}}
\newcommand\1{|\hskip-1.5pt|}

\newcommand{\sL}{\mathscr{L}}

\newcommand{\cC}{\mathcal{C}}
\newcommand{\cH}{\mathcal{H}}

\newcommand{\cL}{\mathcal{L}}
\newcommand{\cN}{\mathcal{N}}

\newcommand{\Spec}{\operatorname{Spec}}
\newcommand{\sign}{\operatorname{sign}}
\newcommand{\Tr}{\operatorname{Tr}}
\newcommand{\Ran}{\operatorname{Ran}}

\newcommand{\mes}{\operatorname{mes}}

\newcommand{\supp}{\operatorname{supp}}

\newcommand{\blangle}{{\boldsymbol{\langle}}}
\newcommand{\brangle}{{\boldsymbol{\rangle}}}

\newcommand{\TF}{\mathsf{TF}}
\newcommand{\D}{\mathsf{D}}
\newcommand{\sfH}{\mathsf{H}}
\newcommand{\x}{\mathsf{x}}
\newcommand{\y}{\mathsf{y}}

\newcommand{\bR}{\mathbb{R}}
\newcommand{\bC}{\mathbb{C}}
\newcommand{\bS}{\mathbb{S}}

\newcommand{\sC}{\mathscr{C}}
\newcommand{\sH}{\mathscr{H}}

\title{%
Strong Scott Conjecture}

\author{Victor Ivrii\thanks{This research was supported in part by National Science and Engineering  Research Council (Canada) Discovery Grant  RGPIN 13827}}
\begin{document}

\maketitle

\begin{abstract}
In heavy atoms  and molecules, on the distances $ a \ll Z^{-1/2}$ from one of the nuclei (with a charge $Z_m$), we prove that the ground state electronic density $\rho_\Psi (x)$ is approximated in $\sL^p$-norm  by the ground state electronic density for a single atom in the model with no interactions between electrons. We cover both non-relativistic and   relativistic cases.
\end{abstract}

\chapter{Introduction}
\label{sect-1}
This paper is a result of my rethinking of one rather old but still remarkable article \cite{IaLS}, which I discovered recently and in which  the asymptotic of the  ground state electronic density on the distances $O(Z^{-1})$ from the nuclei is derived. While there are rather precise results about ground state energy, excessive charges and ionization energy asymptotics\footnote{\label{foot-1} See f.e. \cite{monsterbook}, (in particular, Chapter~\ref{monsterbook-sect-25}, Volume V),  and \cite{ivrii:relativistic-1} in non-relativistic and relativistic cases respectively.}   there are relatively few rigorous results about related electronic density. 

The purpose of this paper is to provide a more refined asymptotics (with an error estimate in $\sL^p$-norms) and on the rather small distances   from the nuclei; now we also cover relativistic case. Larger distances are covered by \cite{ivrii:el-den-2} which combines microlocal and functional-analytical methods rather than uses purely functional-analytical  methods as this paper.

Multielectron Hamiltonian is given by
\begin{gather}
\mathsf{H}=\mathsf{H}_N:=   \sum_{1\le j\le N} H _{V,x_j}+\sum_{1\le j<k\le N}\frac{\mathsf{e}^2}{|x_j-x_k|}
\label{eqn-1.1}\\
\shortintertext{on}
\mathfrak{H}= \bigwedge_{1\le n\le N} \mathscr{H}, \qquad \mathscr{H}=\mathscr{L}^2 (\mathbb{R}^3, \mathbb{C}^q)\simeq 
\mathscr{L}^2 (\mathbb{R}^3\times \{1,\ldots,q\},\mathbb{C})
\label{eqn-1.2}\\
\shortintertext{with}
H_V =T - \mathsf{e} V(x),
\label{eqn-1.3}
\end{gather}
describing $N$ same type particles (electrons) in the external field with the scalar potential $-V$  and repulsing one another according to the Coulomb law; $\mathsf{e}$ is the charge of the electron, $T$ is an \emph{operator of the kinetic energy}.

In the non-relativistic framework this operator is defined as
\begin{align}
&T= \frac{1}{2\mu} (-i\hbar \nabla)^2.
\label{eqn-1.4}
\end{align}
In the relativistic framework this operator is defined as
\begin{align}
&T= \Bigl(c^2 (-i\hbar\nabla)^2 +\mu ^2 c^4\Bigr)^{\frac{1}{2}}-\mu^2c^4 .
\label{eqn-1.5}
\end{align}
Here
\begin{gather}
V(x)=\sum _{1\le m\le M}  \frac{Z_m \mathsf{e}}{|x-\mathsf{y}_m|}
\label{eqn-1.6}\\
\shortintertext{and}
d=\min _{1\le m<m'\le M}|\mathsf{y}_m-\mathsf{y}_{m'}|>0.
\label{eqn-1.7}
\end{gather}
where $Z_m\mathsf{e}>0$ and $\mathsf{y}_m$ are charges and locations of nuclei, $\mu$ is  the mass of the electron.

It is well-known that the non-relativistic operator is always semibounded from below. On the other hand, it
 is also well-known \cite{Herbst, Lieb-Yau} that it this is not necessarily true for the relativistic operator:
 
\begin{remark}\label{rem-1.1}
One particle relativistic  operator is semibounded from below if and only if
\begin{equation}
Z_m \beta \le \frac{2}{\pi}\qquad \forall m=1,\ldots, M;\qquad \beta\coloneqq \frac{\mathsf{e}^2}{\hbar c}.
\label{eqn-1.8}
\end{equation}
\end{remark}
We will assume (\ref{eqn-1.8}), sometimes replacing it by a strict inequality:
\begin{equation}
Z_m \beta \le \frac{2}{\pi}-\epsilon\qquad \forall m=1,\ldots, M;\qquad \beta\coloneqq \frac{\mathsf{e}^2}{\hbar c}.
\label{eqn-1.9}
\end{equation}
We also assume that $d\ge CZ^{-1}$. 

\begin{remark}\label{rem-1.2}
\begin{enumerate}[label=(\roman*), wide, labelindent=0pt]
\item\label{rem-1.2-i}
In the non-relativistic theory by scaling with respect to the spatial and energy variables we can make $\hbar=\mathsf{e}=\mu=1$ while $Z_m$ are preserved.

\item\label{rem-1.2-ii}
In the relativistic theory by scaling with respect to the spatial and energy variables we can make $\hbar=\mathsf{e}=1$, $\mu=\frac{1}{2}$  while $\beta$ and $Z_m$ are preserved.

\item\label{rem-1.2-iii}
In the corresponding single nuclei theory without electron-to-electron interaction (which is used in the proof of Strong Scott conjecture) only $\beta Z$ is preserved and we can assume that $Z=1$.\end{enumerate}
From now on we assume that such rescaling was already made and we are free to use letters $\hbar$, $\mu $ and $c$ for other notations.
\end{remark}

\begin{theorem}\label{thm-1.3}
Assume that
\begin{gather}
\min_{1\le m< m' \le M} |\y_{m}-\y_{m'}|\ge Z^{-1/3+\sigma}
\label{eqn-1.10}\\
\intertext{with $\sigma\ge 0$, and let $U\in \sL^\infty$ such that}
\supp(U)\subset B(\y_m, a),\quad a\le Z^{-1/2-\varkappa},\quad |U|\le 1.
\label{eqn-1.11}
\end{gather}
with arbitrarily small $\varkappa>0$.

Then
\begin{multline}
|\int U\bigl(\rho_\Psi-\rho_{m,\beta}\bigr)\,dx|\\
\le C F^{1/2} \bigl( (Za)^{3/2}\|\langle Z(x-\y_m)\rangle^{-3/2}U\|_{\sL^1}\bigr)^{1/2}+C G
\label{eqn-1.12}
\end{multline}
with
\begin{gather}
F=   \bigl(Z^{13/6 -\delta} a^{-1/2}+  Z^{4}a^{3}\bigr)Z^{\varkappa}
\label{eqn-1.13}\\
\shortintertext{and}
G=  \bigl(Z^{7/6-\delta}a^{3/2} + Z^2 a^3\bigr)Z^{\varkappa},
\label{eqn-1.14}\\
\shortintertext{where}
\rho_{m,\beta} (x)= q Z_m^3  \bar{\rho} _{Z_m\beta} (Z_m (x-\y_m)),
\label{eqn-1.15}
\end{gather}
$\delta =\delta(\sigma)>0$ for $\sigma>0$ and $\delta=0$ for $\sigma=0$, $\bar{\rho}_\beta (x)=\bar{e}_\beta (x,x,0)$, 
$\bar{e}_\beta (x,y,\tau)$ is the Schwartz kernel of the spectral projector $\uptheta (\tau - H_{V^0})$ for toy-model operator  
$H_{V^0}\coloneqq H_{\beta, V^0}=T_\beta -|x|^{-1}$ in $\sL^2(\bR^3,\bC)$ with
\begin{gather}
T_\beta= \left\{\begin{aligned}
&\Bigl( -\beta^{-2}\Delta  +\frac{1}{4}\beta^{-4} \Bigr)^{\frac{1}{2}}-\frac{1}{2}\beta^{-2}&&\beta >0,\\
&-\Delta   &&\beta=0.
\end{aligned}\right.
\label{eqn-1.16}
\end{gather}
Here and below $\langle x\rangle =(|x|^2+1)^{1/2}$.
\end{theorem}

\begin{corollary}\label{cor-1.4}
Let \underline{either} $a\asymp Z^{-1}$ and $X\subset B(\y_m,a)$ \underline{or} ${Z^ {-1}\le a \le Z^{-1/2-\varkappa}}$ and
$X\subset B(\y_m,a)\setminus B(\y_m,a/2)$. Then
\begin{enumerate}[label=(\roman*), wide, labelindent=0pt]
\item\label{cor-1.4-i}
The following estimate holds:
\begin{gather}
\|\rho_\Psi -\rho_{m,\beta}\|_{\sL^1 (X)} \le C F^{1/2} (\mes(X))^{1/2} + CG.
 \label{eqn-1.17}
\end{gather}

\item\label{cor-1.4-ii}
Assume that $|\rho_\Psi |\le \omega $ in $X$ with $\omega  \ge Z^{3/2}a^{-3/2}$\,\footnote{\label{foot-2} Such estimates could be found in \cite{ivrii:el-den}: $\omega  = Z^3$ for $a\le Z^{-8/9}$, $\omega =Z^{19/9} a^{-1}$ for $Z^{-8/9}\le a\le Z^{-7/9}$ and $\omega =Z^{197/90}a^{-9/10}$ for $Z^{-7/9}\le a\le Z^{-1/3}$.}. Then for $p=2,3,\ldots$ the following estimate holds:
\begin{multline}
  \| \rho_\Psi -\rho_{m,\beta} \|_{\sL^p (X)}
  \le \\[4pt]
 \shoveright{ C \omega ^{1- 2(1-2^{-p})/p} \Bigl(F ^{(1-2^{-p})/p} (\mes(X)) ^{2^{-p}/p}   +
   F ^{(1-2^{1-p})/p} G^{1-2^{1-p}}\Bigr)  }\\+ C\omega^{1-1/p}G^{1/p}.
\label{eqn-1.18}
\end{multline}
\end{enumerate}
\end{corollary}

\begin{remark}\label{rem-1.5}
\begin{enumerate}[label=(\roman*), wide, labelindent=0pt]
\item\label{rem-1.5-i}
$F$ is defined by its first term, as $a\le Z^{-11/21-2/\delta/7}$ and by its second term otherwise.
$G$ is defined by its first term as $a\le Z^{-5/9-2\delta/3}$ snd by its second term otherwise.

\item\label{rem-1.5-ii}
Current approximation definitely is not rigt for $a\ge Z^{-1/3}$. On the other hand, resztriction $a\le Z^{-1/2-\varkappa}$ is due to requirement of Proposition~\ref{prop-3.8}. If we do not use this proposition, but all arguments of Section~\ref{sect-4} prior to its use, we get estimate (\ref{eqn-1.12}) with
\begin{gather}
F= \bigl(Z^{13/6-\delta}+a^{-1/2} +Z^{7/2}a^{3/2}\bigr)Z^{\varkappa}
\label{eqn-1.19}\\
\shortintertext{and}
G=(Z^{7/6-\delta}a^{3/2}+Z^{5/2}a^{7/2}\bigr)Z^\varkappa
\label{eqn-1.20}
\end{gather}
which for $a\ge Z^{-1/2-\varkappa}$ are defined by their second terms.

\item\label{rem-1.5-iii}
On the other hand, for $a\gg Z^{-1}$ makes sense approach of  \cite{ivrii:el-den-2}, which is also based on Proposition~\ref{prop-4.1}, but instead of functional-analytic arguments of Sections~\ref{sect-2} and~\ref{sect-3} and Appendices~\ref{sect-A} and~\ref{sect-B} uses microlocal arguments and $\rho_\Psi$ is approximated by Thomas-Fermi density $\rho^\TF$. Since we are going to remake that paper, using some new ideas of this one, we will compare results of two approaches there. 

\item\label{rem-1.5-iv}
On the other hand, while we believe that $\bar{\rho}_\beta$ is different from $\bar{\rho}_0$ (the same density in the non-relativistic case) we have no proof of this. We hope to do it later and to estimate the difference from below.

\item\label{rem-1.5-v}
In \cite{ivrii:el-den} the upper pointwise estimates of $\rho_\Psi$ has been derived in the non-relativistic case  but except for $Z^{-1}$-vicinities of nuclei they are worse than $C\rho^\TF$. 

\item\label{rem-1.5-vi}
In four previous versions of this paper (v1--v4) we considered  only non-relativistic case. Furthermore, first three versions (v1--v3) contain several grave errors which were corrected in v4; moreover, in v4 results were improved for $a$ not much larger than $Z^{-1}$. In the previous version (v5) some improvements were made and proofs were simplified.

\item\label{rem-1.5-vii} Current version (v6) introduced several new ideas and results were significantly improved. 
\end{enumerate}
\end{remark}

\begin{remark}\label{rem-1.6}
Obviously 
\begin{equation} 
\bar{\rho}_\beta(x)=\frac{1}{4\pi}\sum_{n\ge 1}\sum_{0\le l\le n-1} (2l+1)R^2_{n,l}(|x|;\beta)
\label{eqn-1.21}
\end{equation}
with $R_{n,l}(r;\beta)$ defined by (\ref{eqn-B.5}); in particular,
\begin{equation}
\bar{\rho}_0 (0)=\frac{1}{4\pi} \sum_{n\ge 1} \frac{1}{n^3}.
\label{eqn-1.22}
\end{equation}
\end{remark}

\emph{Plan of the paper}. In Sections~\ref{sect-2} and~\ref{sect-3} we consider a one-particle Hamiltonian (relativistic or not) with a  potential 
$V=V^0+\varsigma U$ where $U$ is supported in $B(0, r)$ and satisfies $|U(x)|\le  1$, $0<\varsigma r\le \epsilon$  and explore its eigenvalues and projectors, correspondingly. 

In Section~\ref{sect-4}   we prove Proposition~\ref{prop-4.1}, using arguments of~\cite{IaLS},  and then, using results of Section~\ref{sect-2} and~\ref{sect-3}, we derive from it Theorem~\ref{thm-1.3}. 

 In Appendix~\ref{sect-A} we estimate eigenfunctions of the non-relativistic hydrogen Hamiltonian and in Appendix~\ref{sect-B} we generalize them to the relativistic hydrogen Hamiltonian.
\chapter{Estimates of eigenvalues}
\label{sect-2}

Now we compare negative spectra  of operators $H^0\coloneqq H_{V^0}$ 
and $H\coloneqq H_{V}$ where $H_V=T -V$,  $V^0=|x|^{-1}$ and $V=V^0+\varsigma U$ with 
\begin{gather}
\supp (U) \subset B(0,r), \qquad r\ge 1, \qquad |U|\le1
\label{eqn-2.1}\\
\shortintertext{and}
0<\varsigma r\le \epsilon_0
\label{eqn-2.2} 
\end{gather}
with sufficiently small constant $\epsilon_0$. Then both of these negative spectra are discrete.  

According to Corollary~\ref{cor-B.7} the negative spectrum of  $H^0_V$ consists of the clusters  $\cC_n^0=\{\lambda^0_{n,k} (\beta)\}_{k=1,\ldots, \nu_k}$\,\footnote{\label{foot-3} More precisely, in the non-relativistic case ($\beta=0$) $\lambda^0_{n,k}= \bar{\lambda}_n=-\frac{1}{n^2}$ of multiplicity $n^2$, in the relativistic case for $0<\beta\le \beta_1$  there are also clusters 
 $\cC_n^0$ of the width $O(\beta^2  n^{-3})$, containing $n^2$ eigenvalues and separated by the gaps of the width $\asymp n^{-3}$ and for $\beta_1<\beta \le \frac{2}{\pi}$ separation into clusters is not unique but allows $O(1)$ number of elements moved to or from the neighbouring clusters, so each cluster contains $\nu_n=n^2+O(1)$ eigenvalues (see Remark~\ref{rem-B.8}\ref{rem-B.8-i}). All ``$O$'' are uniform with respect to $n, k$ and $\beta$.}. We denote by $u_{n,k}=u_{n,k}(r,\phi,\theta)$ corresponding eigenfunctios. 

\begin{claim}\label{eqn-2.3}

Let us number eigenvalues of $H_V$ correspondingly by $\lambda_{n,k}$ with  with $n=1,2,\ldots $ and $k=1,2,\ldots, \nu_n$, so that $\lambda_{n,k}\le \mu_{n',k'}$ for  all $n<n'$  and for all corresponding $k$ and $k'$ and $\lambda^0_{n,k} \le   \lambda^0_{n,k'}$ for all $k<k'$.
\end{claim}

\begin{proposition}\label{prop-2.1}
Let \textup{(\ref{eqn-2.1})} and \textup{(\ref{eqn-2.2})} be fulfilled and let us number eigenvalues of $H_V$ and $H_{V^0}$ according to \textup{(\ref{eqn-2.3})}. Furthermore, assume that
\begin{equation}
\varsigma r^{5/2}\le \epsilon n^3.
\label{eqn-2.4}\
\end{equation}
\begin{enumerate}[label=(\roman*), wide, labelindent=0pt]
\item\label{prop-2.1-i}
Then the following estimate holds
\begin{gather}
|\lambda_{n,k}-\lambda^0_{n,k}| \le 
C\left\{\begin{aligned}
&\varsigma r^{5/2} n^{-5} + n^{-3}&&\text{for\ \ } r\le C_0 n^2,\\
&\varsigma + n^{-3} &&\text{for\ \ } r\ge C_0 n^2,
\end{aligned}\right.
\label{eqn-2.5}
\end{gather}
\item\label{prop-2.1-ii}
In particular, $-\lambda_{n,k}\asymp n^{-2}$.
\end{enumerate}
\end{proposition}

\begin{proof}
Indeed, for $U$, which is smooth in $r$-scale, this estimate follows from the standard semiclassical asymptotics of the eigenvalue counting function:
\begin{gather}
\N (\tau ) =\cN^\W (\tau ) + O(|\tau |^{-2})\qquad \text{for\ \ } \tau <0
\label{eqn-2.6}\\
\shortintertext{with}
\cN^\W (\tau )  =\frac{1}{6\pi^2}\int  (\tau -V)_+^{3/2} \,dx\asymp|\tau |^{-3}
\label{eqn-2.7}\\
\intertext{for both $H_V$ and $H_{V^0}$ and}
|\cN^\W (\tau ; H_V) -\cN^\W (\tau ;H_{V^0}) |\le C\varsigma r^{5/2}
\label{eqn-2.8}
\end{gather}
for $r\le C_0|\tau |^{-1}$; otherwise in the right-hand expression we need to replace $r$ by $n^2$.

Then using the standard variational arguments we can drop the smoothness condition.
\end{proof}

\begin{proposition}\label{prop-2.2}
Let condition \textup{(\ref{eqn-2.1})} be fulfilled and 
\begin{equation}
\varsigma r^{3/2}  \le \epsilon _0
\label{eqn-2.9}
\end{equation}
with sufficiently small constant $\epsilon_0$. Then for  $n\ge C_0^{-1}\sqrt{r}$

\begin{enumerate}[label=(\roman*), wide, labelindent=0pt]
\item\label{prop-2.2-i}
The  spectrum of $H_{V}$ consists of the clusters 
$\cC_n=\{\lambda_{n,k}\}_{k=1,\ldots,\nu_n}$ of the width $\le C_0(\beta^2+\varsigma) n^{-3}$ separated by gaps of the width $\asymp n^{-3}$. Each cluster $\cC_n$ contains exactly the same number of eigenvalues as $\cC^0_n$ and they have common borders: both are contained in $[\eta_n+\epsilon_1 n^{-3} , \eta_{n+1}-\epsilon_1 n^{-3}]$.

In particular, in the non-relativistic case and in the relativistic case with $\beta <\epsilon_1$ with sufficiently small constant $\epsilon_1$ (see Corollary~\ref{cor-B.7}\ref{cor-B.7-i}) 
\begin{gather}
\cC_n \subset [\bar{\lambda}_n - C_0(\beta^2+\varsigma) n^{-3}, \bar{\lambda}_n - C_0(\beta^2+\varsigma) n^{-3}].
\label{eqn-2.10}
\end{gather}

\item\label{prop-2.2-ii}
Furthermore, 
\begin{gather}
|\sum_k \bigl(\lambda_{n,k} -\lambda^0_{n,k}\bigr)| \le  C \varsigma r \| \langle x\rangle ^{-3/2} U    \|_{\sL^1} n^{-3}
\label{eqn-2.11}
\end{gather}
\end{enumerate}
where here and below $ \langle x\rangle=(|x|^2+1)^{1/2}$.
\end{proposition}

\begin{proof}
We need to  estimate $\lambda_{n,k}$ from above and below. 

\begin{enumerate}[label=(\alph*), wide, labelindent=0pt]

\item\label{pf-2.2-a}
We start from the easier upper estimate for $\lambda _{n,k}$.  
Let $\N(\tau)$ be the number of eigenvalues below 
$\tau\in [\lambda_{n,k} ^0, \lambda_{n,k} ^0 +C_0n^{-3}]$; we know from Proposition~\ref{prop-2.1}\ref{prop-2.1-ii} that 
under assumption (\ref{eqn-2.9}) $\lambda_{n,k}\le \lambda^0_{n,k}+C_0n^{-3}$ for sure. Recall that
\begin{gather}
\N(\tau) = \max \dim \cL\,,
\label{eqn-2.12}\\
\intertext{where maximum is taken over subspaces $\cL\subset \sL^2(\bR^3)$ such that}
(H_V u,u) - \tau \|u\|^2<0 \qquad \forall u\in \cL, \ u\ne 0. 
\label{eqn-2.13}\\
\intertext{Therefore we can replace the latter condition by}
(H_{V^0} u,u) -\tau \|u\|^2 + \varsigma (U_- u,u)\|^2<0 \qquad \forall u\in \cL, \ u\ne 0\,,
\label{eqn-2.14}
\end{gather}
(where here and below   $U_\pm \coloneqq \max (\pm U,0)$) but then instead of equality in (\ref{eqn-2.12}) we get an inequality
\begin{equation}
\N(\tau)\ge  \max \dim \cL.
\tag*{$\textup{(\ref{eqn-2.12})}_<$}\label{eqn-2.12-<}
\end{equation}

Let  $\lambda=\lambda^0_{n,k}$; let us consider up $\cL= \cL_1\oplus \cL_2 $ with  $\cL_1=\oplus \sum_{n' <n_*} \sH_{n',l'}$, $\cL_2\subset \oplus \sum_{n_*\le n' \le n} \sH_{n',l'}$, where $\sH_{n',l'}$  is a linear span of  basis functions\footnote{\label{foot-4} In the spherical coordinates $(r,\varphi,\theta)$. Here and below $R_{n,l}(r)$ are defined by (\ref{eqn-A.2}), (\ref{eqn-B.5}) in the non-relativistic and relativistic case respectively, $Y_{l',m}(\varphi, \theta)$ are spherical harmonics and a natural range is $m=-l',\ldots, l'$.}  
\begin{equation}
\Upsilon_{n',l',m}\coloneqq \Upsilon_{n',l',m}(r,\varphi,\theta;\beta) =  R_{n',l'}(r;\beta)Y_{l',m}(x/|x|)
\label{eqn-2.15}
\end{equation}
with $m=-l',\ldots,l'$. Let  $n_*=n-C_1$. Then (\ref{eqn-2.14}) holds provided
\begin{multline}
\sum_{1\le i\le 2} \Bigl[(H_{V^0} u_i,u_i) -\tau \|u_i\|^2 + 2\varsigma (U_-u,u) \Bigr]<0 \\
\forall u_i\in \cL_i,\ u=u_1+u_2 \ne 0.
\label{eqn-2.16}
\end{multline}

In virtue of Proposition~\ref{prop-2.1}\ref{prop-2.1-ii} under assumption (\ref{eqn-2.9}) the left-hand expression of (\ref{eqn-2.16}) is negative with any $\tau\ge  \lambda^0_{n,k}$ for all $u=u_1\in\cL_1$, $u\ne 0$.  Then we need to consider only $u_2\in \cL_2$. 
Decomposing $\cL_2\ni v=\sum_{n',l'} R_{n',l'}(|x|)  v _{n',l'}$, where $v _{n',l'}=v _{n',l'}(x/|x|)$ are eigenfunctions of the Laplacian on $\bS^2$ corresponding to eigenvalue $-l'(l'+1)$, we estimate 
\begin{multline}
(U_- u,u) \le (\phi u,u) \\
\begin{aligned}
&\le  \sum_{n',n'';l'} \1 v_{n',l'}\1\, \1v_{n'',l'}\1  \int \phi (t) |R_{n',l'}(t)|\, |R_{n'',l'}(t)|\,t^2dt\\
&\le C \sum_{n',l'} \1v_{n',l'}\1^2  \int \phi (t) |R_{n',l'}(t)|^2\,t^2dt
\end{aligned}
\label{eqn-2.17}
\end{multline}
with $\phi(x)\coloneqq\phi_r(x)$, characteristic function of $B(0,r)$; we used that $\phi(x)$ is spherically symmetric, that  $w_{n',l'}$ and $w_{n'',l''}$ are 
orthogonal for $l'\ne l''$, and that $n'$ and $n''$ are between $n^*=n-C_1$ and $n$. Here $\1v\1=\1v\|_{\sL^2(\bS^2)}$.

Recall that due to Propositions~\ref{prop-A.8}\ref{prop-A.8-iii} and \ref{prop-B.5} in the non-relativistic and relativistic cases respectively,\begin{gather}
|R_{n',l}|\le C\left\{\begin{aligned}
 &\langle x\rangle^{-3/4}n'^{-3/2} &&\text{for\ \ } l\le C_0\langle x\rangle^{1/2},\\
 &l^{-s} n'^{-3/2} &&\text{for\ \ } l\ge C_0\langle x\rangle^{1/2}.
 \end{aligned}\right.
 \label{eqn-2.18}\\
\shortintertext{Then}
(U_- u,u) \le C r^{3/2}  n^{-3} \sum_{n',l'}\1 v_{n',l'}\1^2 =  C r^{3/2}  n^{-3}\|u\|^2 \qquad \text{for\ \ } u\in \cL_2\,.
\label{eqn-2.19}
\end{gather}

Therefore  we can take $n_*=n$: \begin{phantomequation}\label{eqn-2.20}\end{phantomequation}\begin{phantomequation}\label{eqn-2.21}\end{phantomequation}
\begin{gather}
\lambda_{n', k'}\le \lambda_{n,k}^0\qquad  \text{for\ \ } n'<n,\ k,\ k'.
\tag*{$\textup{(\ref{eqn-2.20})}_<$}\label{eqn-2.20-<}\\
\intertext{Further, we conclude that}
\lambda_{n,k} \le \mu_{n,k}
\tag*{$\textup{(\ref{eqn-2.21})}_<$}\label{eqn-2.21-<}
\end{gather}
 which are eigenvalues of $\uppi_n^0 H_{V^0+ 2\varsigma U_-}\uppi_n^0$ where $\uppi_n^0$ is a projector to
$\cL_2=\oplus \sum_l \sH_{n,l}$. Therefore\begin{phantomequation}\label{eqn-2.22}\end{phantomequation}\begin{phantomequation}\label{eqn-2.23}\end{phantomequation}
\begin{gather}
\sum_{k} \lambda_{n,k} \le \Tr \bigl[\uppi_n^0 H_{V^0- 2\varsigma U_-}\uppi_n^0\bigr]=
\Tr \bigl[\uppi_n^0 H_{V^0}\uppi_n^0\bigr] + 2\varsigma \Tr \bigl[\uppi_n^0 U_-\uppi_n^0\bigr]
\tag*{$\textup{(\ref{eqn-2.22})}_<$}\label{eqn-2.22-<}
\end{gather}
where the first term equals $\sum_k \lambda_{n,k}^0$ and 
\begin{multline}
\Tr \bigl[ \uppi_n^0U_-\uppi_n^0\bigr] =\sum_{l,m} \int U_-(x) |R_{n,l}(|x|)|^2 |Y_{l,m} (x/|x|)|^2\,dx \\
=\frac{1}{4\pi} \sum_{l,m}(2l+1)  \int U_-(x) |R_{n,l}(|x|)|^2 \,dx \le Cr \| \langle x\rangle U_-\|_{\sL^1} n^{-3}
\tag*{$\textup{(\ref{eqn-2.23})}_<$}\label{eqn-2.23-<}
\end{multline}
due to equality
\begin{gather}
\sum_m |Y_{l,m}|^2 = \frac{1}{4\pi}(2l+1)
\label{eqn-2.24}
\end{gather}
and (\ref{eqn-2.18}). We arrive to
\begin{gather}
\sum_{k} (\lambda_{n,k}- \lambda_{n,k} ^0) \le  C\varsigma r\|\langle x\rangle^{-3/2} U_-\|_{\sL^1}.
\tag*{$\textup{(\ref{eqn-2.11})}_<$}\label{eqn-2.11-<} 
\end{gather}

\item\label{pf-2.2-b}
To estimate $\lambda_{n,k}$ from below we need to estimate $\N(\tau)$ from above.
To do this we replace (\ref{eqn-2.16}) by 
\begin{multline}
\sum_{1\le i\le 2} \Bigl[(H_{V^0} u_i,u_i) -\tau \|u_i\|^2 - 2\varsigma (U_+u_i,  u_i)\Bigr]<0 \\
\forall u_i\in \cH_i,\ u=u_1+u_2\ne 0. 
\label{eqn-2.25}
\end{multline}
Then $\cH_1$ and $\sH_2$ become ``detached''. Let us take $\cH_1=\oplus \sum_{n' \le n^*, l'} \sH_{n',l'}$ with $n^*=n+  C_1$,  and $\cH_2 = \sL^2(\bR^3)\ominus \cH_1$.  Then in virtue of Proposition~\ref{prop-2.1} for $u=u_2\in \cH_2$ expression (\ref{eqn-2.25}) is non-negative and therefore (\ref{eqn-2.12}) holds with maximum taken over $\cL\subset \cH_1$.

Note that  
\begin{gather}
\N(\tau) = N- \tilde{\N}(\tau)\qquad \text{with\ \ } \tilde{\N}(\tau) =\max \dim \tilde{\cL}\,,
\label{eqn-2.26}
\end{gather}
where maximum is taken over subspaces on which quadratic form 
\begin{equation}
(H_{V^0}u,u)-\tau \|u\|^2+2\varsigma (U_+u,u) \ge 0
\label{eqn-2.27}
\end{equation}
and $N=\dim \cH_1$. Now we need to estimate $\tilde{N}(\tau)$ from below can pick-up test functions delivering this inequality. We consider
$\tilde{\cL} \subset \oplus \sum_{n\le n'\le n^*, l'} \sH_{n',l'}$.

Then, repeating arguments of  Part~\ref{pf-2.2-a}   we arrive to 
\begin{gather}
\lambda_{n', k'} > \lambda_{n,k}^0\qquad  \text{for\ \ } n' > n,\ k,\ k'.
\tag*{$\textup{(\ref{eqn-2.20})}_>$}\label{eqn-2.20->}\\
\intertext{Further, we conclude that}
\lambda_{n,k} \ge \mu'_{n,k}
\tag*{$\textup{(\ref{eqn-2.21})}_>$}\label{eqn-2.21->}\\
\intertext{which are eigenvalues of $\uppi_n^0 H_{V^0 - 2\varsigma U_+}\uppi_n^0$ and}
\sum_{k} \lambda_{n,k} \le \Tr \bigl[\uppi_n^0 H_{V^0- 2\varsigma U_+}\uppi_n^0\bigr]=
\Tr \bigl[\uppi_n^0 H_{V^0}\uppi_n^0\bigr] - 2\varsigma \Tr \bigl[\uppi_n^0 U_+\uppi_n^0\bigr]
\tag*{$\textup{(\ref{eqn-2.22})}_>$}\label{eqn-2.22->}\\
\intertext{where the first term equals $\sum_k \lambda_{n,k}^0$ and }
\Tr \bigl[ \uppi_n^0U_+\uppi_n^0\bigr]  \le Cr \| \langle x\rangle U_+\|_{\sL^1} n^{-3}
\tag*{$\textup{(\ref{eqn-2.23})}_>$}\label{eqn-2.23->}\\
\intertext{due to equality (\ref{eqn-2.24}) and (\ref{eqn-2.18}). We arrive to}
-\sum_{k} (\lambda_{n,k}-\lambda^0_n)\le C \varsigma r \| \langle x\rangle ^{-3/2+\nu}U_+   \|_{\sL^1} n^{-3}.
\tag*{$\textup{(\ref{eqn-2.11})}_>$}\label{eqn-2.11->}
\end{gather}
We leave details to the reader. 

Combining \ref{eqn-2.20-<} and \ref{eqn-2.20->}  we arrive to  Statement~\ref{prop-2.2-i} and combining 
\ref{eqn-2.11-<} and \ref{eqn-2.11->}  we arrive to   Statement~\ref{prop-2.2-ii}.
\end{enumerate}\vskip-1.6\baselineskip
\end{proof}

\begin{proposition}\label{prop-2.3}
Let conditions \textup{(\ref{eqn-2.1})} and  \textup{(\ref{eqn-2.9})} be fulfilled. Then for  $n\le C_0^{-1}\sqrt{r}$
 the  following estimates hold:
 \begin{align}
&|\lambda_{n,k}-\lambda_{n,k}^0|\label{eqn-2.28}\\
&\qquad\le   C\varsigma \bigl(\| \langle x\rangle^{-3/2}U\|_{\sL^1 (B(0,C_0n^2))} + \| \langle x\rangle^{-s}U\|_{\sL^1 (B(0,r)\setminus B(0,C_0n^2))}\bigr),\notag\\
\shortintertext{and}
&\sum_k  |\lambda_{n,k}-\lambda_{n,k}^0|\label{eqn-2.29}\\
&\ \qquad\le   C\varsigma \bigl(\| \langle x\rangle^{-1/2}U\|_{\sL^1 (B(0,C_0n^2))} + \| \langle x\rangle^{-s}U\|_{\sL^1 (B(0,r)\setminus B(0,C_0n^2))}\bigr) 
\notag\end{align}
with arbitrarily large exponent $s$.
\end{proposition}

\begin{proof}
Proof follows the proof of Proposition~\ref{prop-2.2} and takes in account that 
\begin{gather}
|R_{n,l} (r) |\le C r^{-s} \qquad \text{for\ \ } n\le C_0^{-1}\sqrt{r}.
\label{eqn-2.30}
\end{gather}
We leave easy details to the reader.
\end{proof}

\chapter{Estimates of projectors}
\label{sect-3}

In this section we first consider projectors associated with operators $H_{V^0}$ and $H_V$ and provide some estimate for them; then we combine these results with the results of Section~\ref{sect-2}. We assume that condition~(\ref{eqn-2.9}) is fulfilled. Then due to Proposition~\ref{prop-2.2}\ref{prop-2.2-i}  there is a cluster of eigenvalues  
$\cC^0_n=\{\lambda^0_{n,k}\}_{k=1,\ldots,\nu_n}$ and $\cC_n=\{\lambda_{n,k}\}_{k=1,\ldots,\nu_n}$
of operators $H^0\coloneqq H_{V^0}$ and $H\coloneqq H_{V}$ respectively. Let us denote by $\uppi^0_n$ and $\uppi_n$ the projectors to the corresponding spectral subspaces.  

\begin{proposition}\label{prop-3.1}
\begin{enumerate}[label=(\roman*), wide, labelindent=0pt]
\item\label{prop-3.1-i}
In the framework of Proposition~\ref{prop-2.2}  assume that\footnote{\label{foot-5} Obviously this is a stronger assumption than (\ref{eqn-2.9}).}
\begin{equation}
 \varsigma \bigl(r^{3/2}(1+|\log (rn^{-2})|)+  n^2\bigr) \le \epsilon.
\label{eqn-3.1}
\end{equation}
Then the following estimate holds:
\begin{multline}
\| \uppi_n - \uppi^0_n\| \le
C\varsigma \bigl(r^{3/2}(1+|\log (rn^{-2})|) + n^{1/2}r^{3/4}\bigr)\\
\times \Bigl[ 1+\bigl(\varsigma  \bigl( r^{3/2}  (1+|\log (rn^{-2})|) + n^{2} \bigr)\bigr)^{K-1}\varsigma n^3\Bigr]
 \label{eqn-3.2}
\end{multline}
with arbitrarily large exponent $K$.

\item\label{prop-3.1-ii}
On the other hand, in the framework of Proposition~\ref{prop-2.3} assume that 
$\supp(U)\subset B(0,r)\setminus B(0,r/2)$.
Then the following estimate holds
\begin{equation}
\| \uppi_n - \uppi^0_n\| \le C\varsigma r^{-s}
 \label{eqn-3.3}
\end{equation}
with arbitrarily large exponent $s$.
\end{enumerate}
\end{proposition}

\begin{proof}
Due to Lemma~\ref{lemma-3.2} below it is sufficient to prove estimates (\ref{eqn-3.2}) and (\ref{eqn-3.3}) for operator 
$(\uppi_n-\uppi_n^0)\uppi_n^0$.

\begin{enumerate}[label=(\roman*), wide, labelindent=0pt]
\item\label{pf-3.1-i}
Note that 
\begin{equation}
\uppi_n =\frac{1}{2\pi i} \oint_{\gamma_n} (z- H)^{-1}\,dz\,,
\label{eqn-3.4}
\end{equation}
where $\gamma_n $ is the circle of radius $\asymp n^{-3}$ with with counter-clockwise orientation, 
contains inside cluster of the eigenvalues $\cC_n$  and outside all other clusters and $\dist (\gamma_n, \Spec (H))\asymp n^{-3}$. 
Similar formula holds for $\uppi_n^t$, corresponding to operator $H^t=H_{V^0 +t\varsigma U}$, and we can select a circle which serves to all  $|t|\le 1$. 

Then
\begin{multline}
\uppi_n -\uppi_n^0 =\frac{1}{2\pi i} \oint_{\gamma_n} \bigl[(z- H)^{-1}-(z- H^0)^{-1}\bigr]\,dz=\\
\begin{aligned}
 \sum_{1\le k\le K} 
 -\frac{\varsigma ^k}{2\pi i} &\oint_{\gamma_n} (z- H^0)^{-1}U\cdots U(z- H^0)^{-1}\,dz \\
-\frac{\varsigma ^{K+1}}{2\pi i} &\oint_{\gamma_n}  (z- H)^{-1}U (z- H^0)^{-1}U\cdots U(z- H^0)^{-1} \,dz\,,\qquad
\end{aligned}
\label{eqn-3.5}
\end{multline}
where each \emph{regular term} (with $k\le K$) contains $(k+1)$ factors $(z-H^0)^{-1}$ and $k$ factors $U$, while the \emph{remainder term}  contains $(K+1)$ factors $(z-H^0)^{-1}$, one factor $(z-H)^{-1}$ (not on the last position) and $(K+1)$ factors $U$.  ``Out of the box'' the  each factor $(z-H^0)^{-1}$ or $(z-H)^{-1}$ has an operator norm not exceeding $Cn^3$, so the operator norm of the whole term does not exceed $C(\varsigma n^3)^k$. To improve it, multiply by $\uppi_n^0$ from the right:
\begin{multline}
(\uppi_n -\uppi_n^0)\uppi_n^0  =
 \sum_{1\le k\le K} 
 -\frac{\varsigma ^k}{2\pi i} \oint_{\gamma_n} (z- H^0)^{-1}U\cdots U(z- H^0)^{-1}\uppi_n^0\,dz \\
-\frac{\varsigma ^{K+1}}{2\pi i} \oint_{\gamma_n}  (z- H)^{-1}U (z- H^0)^{-1}U\cdots U(z- H^0)^{-1}\uppi_n^0\,dz\,.
\label{eqn-3.6}
\end{multline}
Observe that $\| U \uppi_n^0\| \le Cr^{3/4}n^{-3/2}$   in  virtue of  (\ref{eqn-2.18}),
so the operator norm of $k$-th term does not exceed $C(\varsigma n^3)^k r^{3/4}n^{-3/2}$. 

To improve this estimate, let us replace   $j$-th copy of $(z-H^0)^{-1}$ (with $j=1,\ldots,k$) by
\begin{equation*}
\uppi^0_{p_j}(z-H^0)^{-1}=\uppi^0_{p _j}(z-H^0)^{-1}\uppi^0_{p _j}
\end{equation*}
 with $p_j\asymp n$. Then each factor $U$ will be sandwiched between $\uppi^0_{p_j}$ and $\uppi^0_{p_{j+1}}$ and $\| \uppi^0_q U \uppi_p^0\| \le Cr^{3/2}$ again in virtue of (\ref{eqn-2.18}). Therefore the operator norm of this new $k$-th term (with $k\le K$) does not exceed   
\begin{gather*}C(\varsigma r^{3/2})^k \prod _{1\le j\le k}  (|p_j-n|+1)^{-1}\\
\shortintertext{because}
\|\uppi^0_p (z-H^0)^{-1}\|\le Cn^3 (|p-n|+1)^{-1}.
\end{gather*}
If we take a sum of this expression  over $p_j\colon |p_j-n|\le P$ we get
\begin{equation}
C(\varsigma r^{3/2})^k(1+ \log (P)) ^k  
\label{eqn-3.7}
\end{equation}
(with $1\le P\le n/2$) and therefore if we insert instead of $\uppi_{p_j}$ their sums 
\begin{equation}
\Pi_{n,P}=\sum_{p\colon |p-n|\le P}\uppi_p^0
\label{eqn-3.8}
\end{equation}
we get that the operator norm would not exceed (\ref{eqn-3.7}) (for $P\ge 1$). 

On the other hand, if instead of $\Pi_{n,P}$ we insert $(I- \Pi_{n,P})$, then the corresponding factor  $\varsigma r^{3/2}(1+\log (P)))$ in the estimate should be replaced by  $\varsigma n^3 P^{-1}$ because $\|(I-\Pi_{n,P} )(z-H^0)^{-1}\|\le Cn^3 P^{-1}$. 

\begin{claim}\label{eqn-3.9}
For  the left-most copy, however of $(z-H^0)^{-1}$, since it  serves just one $U$ (on the right of it), we actually get  $\varsigma (n^{3/2}r^{3/4}  (1+\log(P)) + n^3 P^{-1}) $.
\end{claim}

Therefore the operator norm of the regular  term with $k=1$ in (\ref{eqn-3.6}) does not exceed
$C\varsigma  \bigl( r^{3/2} \log(P) + n^{3/2}r^{3/4} P^{-1}\bigr) $ and optimizing this expression by $P\colon P\le n/2$ we get 
the first factor in the right-hand expression of (\ref{eqn-3.2}), achieved as $P\asymp \min (n^{3/2}r^{-3/4},\, n)$.

Further, the operator norm of the regular  term with $k\ge 2$ in (\ref{eqn-3.6}) does not exceed
$C\varsigma  ^k \bigl( r^{3/2} \log(P) + n^{3/2}r^{3/4} P^{-1}\bigr)   \bigl( r^{3/2} \log(Q) + n^{3} Q^{-1}\bigr) ^{k-1}$
because we can take $\Pi_{n,P}$ for the left-most factor $(z-H^0)^{-1}$ and $\Pi_{n,Q}$ for all other factors $(z-H^0)^{-1}$  and we take 
$P$ the same as above and $Q\asymp \min (n^{3}r^{-3/2},\, n)$. Then we get  the first factor in the right-hand expression of (\ref{eqn-3.2}), multiplied by 
\begin{gather}
\bigl(\varsigma  \bigl( r^{3/2}  (1+|\log (rn^{-2})|) + n^{2} \bigr)\bigr)^{k-1},
\label{eqn-3.10}
\end{gather}
which due to condition (\ref{eqn-3.1}) is less than $1$.

The same procedure, applied to  the remainder term, results in the expression for $k=K$ multiplied by  $\varsigma n^3$, which means the first factor in (\ref{eqn-3.2}), multiplies by expression (\ref{eqn-3.10}) with $k=K$, and multiplied by $\varsigma n^3$. It completes the proof of Statement~\ref{prop-3.1-i}.

\item\label{pf-3.1-ii}
Proof of Statement~\ref{prop-3.1-ii} is trivial due to (properties of $R_{n,l}$)\ref{eqn-2.18}) and assumption $\supp(U)\subset B(0,r)\setminus B(0,r/2)$ with $n\le C_0^{-1}\sqrt{r}$.
\end{enumerate}
\vskip-1.7\baselineskip
\end{proof}

\begin{lemma}\label{lemma-3.2}
Let $\uppi $ and $\uppi^0$ be two orthogonal projectors such that
\begin{gather}
\|(I-\uppi) \uppi^0\|\le \varepsilon <1/2
\label{eqn-3.11}\\
\shortintertext{and}
\dim \Ran (\uppi)= \dim \Ran (\uppi^0)<\infty.
\label{eqn-3.12}
\intertext{Then for all $p\colon 1\le p\le \infty$}
\|\uppi-\uppi^0\|_p\le 2(1-\varepsilon)^{-1}\|(I-\uppi) \uppi^0\|_p.
\label{eqn-3.13}
\end{gather}
\end{lemma}

\begin{proof}
Due to assumption (\ref{eqn-3.11})  operator $\uppi:\Ran (\uppi^0)\to \Ran(\uppi)$  is injective and therefore due to (\ref{eqn-3.12}) it is invertible, and the norm of an inverse operator $\upvarphi :\Ran(\uppi) \to \Ran(\uppi^0)$ does not exceed $(1-\varepsilon)^{-1}$ and also
\begin{gather*}
\|(I- \upvarphi)\uppi\|_p \le (1-\epsilon)^{-1} \|\uppi-\uppi^0\|_p.
\end{gather*}
Then the same estimate holds for the corresponding norm of 
\begin{gather*}
\uppi^0(I- \upvarphi)\uppi =\uppi^0\uppi - \uppi^0 \upvarphi\uppi =\uppi^0\uppi -\uppi
\end{gather*}
and for adjoint  to it:   $\|\uppi \uppi^0 -\uppi\|_p\le (1-\epsilon)^{-1} \|\uppi-\uppi^0\|_p$.
Since 
\begin{gather*}
\| \uppi^0 -\uppi\|_p\le \|\uppi \uppi^0 -\uppi\|_p + \|\uppi \uppi^0 -\uppi^0\|_p
\end{gather*}
 we arrive to (\ref{eqn-3.18}).
\end{proof}

\begin{proposition}\label{prop-3.3}
\begin{enumerate}[label=(\roman*), wide, labelindent=0pt]
\item\label{prop-3.3-i}
In the framework of Proposition~\ref{prop-2.2}  the following estimates hold 
\begin{gather}
|\Tr [U\uppi^0_n]|\le Cr\|\langle x\rangle^{-3/2}\|_{\sL^1}n^{-3}
\label{eqn-3.14}
\end{gather}
and under assumption \textup{(\ref{eqn-3.1})} also
\begin{multline}
|\Tr [U(\uppi_n-\uppi^0_n)]|\\
\shoveright{\le 
C\varsigma r \|\langle x\rangle^{-3/2}\|_{\sL^1} \bigl( r^{3/2}n^{-3}(1+|\log (rn^{-2})|) +    n^{-1}\bigr)}\\[3pt]
+ C\varsigma r^{5/2} \bigl(\varsigma  \bigl( r^{3/2}  (1+|\log (rn^{-2})|) + n^{2} \bigr)\bigr)^{K-1}\varsigma n^3\Bigr].
\label{eqn-3.15}
\end{multline}
\item\label{prop-3.3-ii}
In the framework of Proposition~\ref{prop-3.1}\ref{prop-3.1-ii}   the following estimate holds
\begin{gather}
|\Tr [U\pi^0_n]|+ |\Tr [U\pi _n]|\le  C r^{-s}.
\label{eqn-3.16}
\end{gather}
\end{enumerate}
\end{proposition}

\begin{proof}
We  actually proved (\ref{eqn-3.14}) (see \ref{eqn-2.23-<} and \ref{eqn-2.23->}).

To prove (\ref{eqn-3.15}) we apply the same approach as in the proof of (\ref{eqn-3.2}) but we insert factors   $\uppi_{m_j}^0$ or $(I-\Pi_{n,Q})$ after the last copy of $(z-H_0)^{-1}$ and, due to the trace, the same factor appears in front of the first copy of $U$. So again, each factor $U$ is sandwiched between two factors of of the type $\uppi_m^0$ or $(I-\Pi_{n,Q})$. 

Observe that in the regular terms one of the factors must be $\uppi_n^0$; otherwise the integrand is holomorphic inside of the contour  of integration $\gamma_n$ and an integral is $0$. And since there is one factor $\uppi_n^0$, the second appears automatically (on both sides of the same copy of $(z-H^0)^{-1}$. 

Therefore two estimate this term for two copies of $U\uppi_n^0$ and $\uppi_n^0 U$ we can use Hilbert-Schmidt norms. But
\begin{equation}
\| U\uppi_n^0\|_2^2 = \| \uppi_n^0U\uppi_n^0\|_2^2 =\Tr [\uppi_n^0 |U|^2 \uppi_n^0]\le C \| \langle x\rangle^{-3/2}U^2\|_{\sL^1}
\label{eqn-3.17}
\end{equation} 
and due to assumption $|U|\le 1$ we can replace there $U^2$ by $U$. So in comparison with the norm estimate each rreqular term acquires two factors $\| \langle x\rangle^{-3/2}U\|_{\sL^1}^{1/2}$ (thus, just one factor $\| \langle x\rangle^{-3/2}U\|_{\sL^1}$).
Then the absolute value of the first term does not exceed
\begin{gather*}
C\varsigma r \| \langle x\rangle^{-3/2}U^\|_{\sL^1} \bigl( r^{3} n^{-3} (1+\log (P)) +  r^{3/2} P^{-1}\bigr)\,,
\end{gather*}
and optimizing by $P\colon P\le n$ we arrive to the first term in (\ref{eqn-3.15}).

The  $k$-th regular term is estimated in the same way albeit with an extra factor (\ref{eqn-3.10}), which is less than $1$. 

While the remainder term could be  estimated a bit better than in (\ref{eqn-3.15}), we do not care, since under proper assumption it will be negligible anyway. The same applies to estimate (\ref{eqn-3.16}) as well.
\end{proof}

\begin{proposition}\label{prop-3.4}
The following formula holds
\begin{equation}
\Tr \bigl[H_V^- - H_{V^0}^-\bigr] = -\int _0^\varsigma \Tr \bigl[U \uptheta (-H_{V^t})\bigr]\,dt.
\label{eqn-3.18}
\end{equation}
with $V^t=V^0 + t U$.
\end{proposition}

\begin{proof}
The proof is trivial.\end{proof}

\begin{remark}\label{rem-3.5}
\begin{enumerate}[label=(\roman*), wide, labelindent=0pt]
\item\label{rem-3.5-i}
One can rewrite the right-hand expression of (\ref{eqn-3.18})  as
\begin{gather}
-\varsigma \Tr \bigl[ U \uptheta (-H_{V^0})\bigr] - 
\int _0^\varsigma \Tr \bigl[U \bigl( \uptheta (-H_{V^t})-\uptheta (-H_{V^0})\bigr)\bigr]\,dt
\label{eqn-3.19}\\
\intertext{with the last term equal}
\sum_{n\ge 1}
-\int _0^\varsigma \Tr \bigl[U ( \uppi_n^t-\uppi_n^0)\bigr]\,dt 
\label{eqn-3.20}
\end{gather}
with $\uppi_n^t$ associated with $H_{V^t}$.

\item\label{rem-3.5-ii}
Further, each term in this sum could be rewritten as
\begin{equation}
\sum_{1\le k \le \nu_n }  (\lambda _{n,k} - \lambda ^0_{n,k}). 
\label{eqn-3.21}
\end{equation}
\end{enumerate}
\end{remark}

Indeed, Statement \ref{rem-3.5-i} is trivial, and to prove Statement \ref{rem-3.5-ii} observe that expression (\ref{eqn-3.21}) is equal to
\begin{align*}
\frac{1}{2\pi i}\Tr \Bigl[ \oint_{\gamma_n} H _{V^t}  (z- H_{V^t } )^{-1}\, dz\Bigr]^{t=\varsigma}_{t=0}=&
\frac{1}{2\pi i}\Tr \Bigl[ \int_0^\varsigma \Bigl( \frac{\partial\ }{\partial t} 
\oint_{\gamma_n} H _{V^t}  (z- H_{V^t } )^{-1}\, dz\Bigr)\,dt\Bigr] \\
=&
\frac{1}{2\pi i} \Tr \Bigl[ \int_0^\varsigma \Bigl( 
-\oint_{\gamma_n} U   (z- H_{V ^t} )^{-1}\, dz     \\
+& \oint_{\gamma_n}  \underbracket {H_{V^t}(z- H_{V ^t} )^{-1}U  (z- H_{V ^t} )^{-1} }\,dz    \Bigr)\,dt \Bigr].
\end{align*}
Observe that because of the trace  we can rewrite the  selected expression  as $H_{V^t}(z- H_{V ^t} )^{-2}U $  and  that $\oint _{\gamma_n} (z- H_{V ^t} )^{-2}\, dz=0$. What remains is obviously equal to the term in (\ref{eqn-3.20}). 

\begin{proposition}\label{prop-3.6}
In the framework of Proposition~\ref{prop-3.1} assume that
\begin{gather}
r\le \varsigma ^{-2/3+\varkappa}
\label{eqn-3.22}
\intertext{with arbitrarily small exponent $\varkappa>0$. Then}
|\Tr \bigl[H_V^- - H_{V^0}^-\bigr] + \varsigma \Tr \bigl[ U \uptheta(-H_{V^0}))\bigr]|\le  
C\varsigma^{2-\varkappa}  r\|\langle x\rangle^{-3/2}U  \|_{\sL^1}.
\label{eqn-3.23}
\end{gather}
\end{proposition}

\begin{proof}
The left-hand expression in (\ref{eqn-3.22}) does not exceed
\begin{multline}
\sum_{n\le N} |\sum_{k\le \nu_n} \bigl(\lambda_{n,k}-\lambda^0_{n,k} \bigr)+\varsigma \Tr [U\uppi^0_n]\, |+\\
\sum _{n>  N} \Bigl( \sum_{k\le \nu_n} \underbracket{|\lambda_{n,k}-\lambda^0_{n,k} |} + |\varsigma \Tr [U\uppi^0_n]|\Bigr)\,,
\label{eqn-3.24}
\end{multline}
where $N$ satisfying (\ref{eqn-3.1}) (if we substitute it instead of $n$) we will chose later.

Then  the terms in the first sum due to  Proposition~\ref{prop-3.3} and Remark~\ref{rem-3.5} do not exceed  the right-hand expression of (\ref{eqn-3.15}), multiplied by $\varsigma$, and their sum does not exceed the first term in the expression below (we sum for $n\colon C_0^{-1}\sqrt{r} \le n \le N$ assuming that $N\ge C_0^{-1}\sqrt{r}$). 

Meanwhile, due to Proposition~\ref{prop-2.2} the terms in the second sum do not exceed 
$C \varsigma r\|\langle x\rangle^{-3/2}U  \|_{\sL^1}n^{-3}$ and their sum does not exceed the second term in the same expression below:
\begin{equation}
C \varsigma   r\|\langle x\rangle^{-3/2}U  \|_{\sL^1}\Bigl( \varsigma   (1+|\log (rN^{-2})|) +  N^{-2}\Bigr).
\label{eqn-3.25}
\end{equation}
We skipped the remainder term which is justified under assumption
\begin{equation}
\varsigma^{1-\varkappa} (r^{3/2} +N^2) \le 1
\label{eqn-3.26}
\end{equation}
with arbitrarily small exponent $\varkappa>0$; then the remainder term does not  exceed $\varsigma^s$. Condition (\ref{eqn-3.26})  is slightly more restrictive than (\ref{eqn-3.1}) and is fulfilled due to (\ref{eqn-3.22}) and the choice $N=\varsigma^{-(1-\varkappa)/2}$, which almost minimizes 
(\ref{eqn-3.25}) but also satisfies  (\ref{eqn-3.26}).\end{proof}

While Proposition~\ref{prop-3.6} will be completely sufficient for our purpose in Section~\ref{sect-4} for small distances from the nuclei, for larger distances we will need a bit more complicated setup. Namely, instead of $V^0$ unperturbed potential will be  $V^0+\varepsilon \Phi$, where $\Phi$ satisfies the same restriction (\ref{eqn-2.1})  as $U$.

First, however, we need to modify Proposition~\ref{prop-3.3}:

\begin{proposition}\label{prop-3.7}
In the framework of Proposition~\ref{prop-2.2} assume that 
\begin{gather}
\supp (\Phi)\subset B(0,r),\qquad |\Phi|\le 1.
\label{eqn-3.27}\\
\shortintertext{and}
(\varepsilon+\varsigma)^{1-\varkappa}(r^{3/2}+n^2) \le 1
\label{eqn-3.28}
\end{gather}
wirh arbitrarily small exponent $\varkappa>0$. Then the following estimate hold:
\begin{multline}
|\Tr \bigl[U (\uppi_{n}(\varepsilon,\varsigma) - \uppi_{n}^0)\bigr]|+ |\Tr \bigl[\Phi  (\uppi_{n}(\varepsilon,\varsigma) - \uppi_{n}^0)\bigr]| \\[3pt]
\le Cr \bigl(\varsigma \| \langle x\rangle^{-3/2}U\|_{\sL^1}   +\varepsilon r^{3/2}\bigr)\bigl( r^{3/2}n^{-3}(1+|\log (rn^{-2})|) +    n^{-2}\bigr)\,,
\label{eqn-3.29}
\end{multline}
where  $\uppi_{n}(\varepsilon,\varsigma)$ are  projectors corresponding  operator to $H_{V^0+\varsigma U+\varepsilon \Phi}$.
\end{proposition}

\begin{proof}
Proof repeats the proof of Proposition~\ref{prop-3.3}. We leave easy details to the reader. 
\end{proof}

Now we can modify Proposition~\ref{prop-3.6}:
\begin{proposition}\label{prop-3.8}
In the framework of Proposition~\ref{prop-2.2}  assume that
\begin{gather}
\varepsilon r^{3/2}\le \epsilon_0
\label{eqn-3.30}
\end{gather}
and \textup{(\ref{eqn-3.27})} is fulfilled. Then the following estimate holds:
\begin{multline}
|\Tr \bigl[ H_{V^0+\varepsilon \Phi}^- -   H_{V^0}^-    -   H_{V^0+\varsigma U+\varepsilon \Phi}^-  +   H_{V^0+\varsigma U}^-\bigr] | \\[3pt]
\le  C(\varsigma+\varepsilon)^{1-\varkappa}  r\bigl(\varsigma \|\langle x\rangle^{-3/2}U  \|_{\sL^1}+  \varepsilon r^{3/2}\bigr).
\label{eqn-3.31}
\end{multline}
\end{proposition}

\begin{proof}
Observe that  the left-hand expression  of (\ref{eqn-3.31}) equals (without traces or absolute value)
\begin{multline}
 U \int_0^\varsigma \Bigl(  \uptheta ( H_{V^0+\varepsilon \Phi +t \varsigma U}) -  \uptheta ( H_{V^0  +tU})\Bigr)\,dt\\
\begin{aligned}
& = U \int_0^\varsigma \Bigl(  \uptheta ( H_{V^0+\varepsilon \Phi +t \varsigma U}) -  \uptheta ( H_{V^0  })  \Bigr)\,dt \\
& -U \int_0^\varsigma \Bigl(  \uptheta ( H_{V^0+\varepsilon \Phi}) -  \uptheta ( H_{V^0 })  \Bigr)\,dt\,,
\end{aligned}
 \label{eqn-3.32}
\end{multline}
and under assumption (\ref{eqn-3.28}) we can break the last expressions into sums of terms with $\uptheta (\ldots)$ replaced by the corresponding projectors.   We also can combine  the first with the second and the third with the fourth terms, rather than 
 the first with the third and the second with the fourth terms. 
 
 Then using the corresponding estimates in Proposition~\ref{prop-3.7} we can estimate the sum over $n\colon n\le N$ by
 the right-hand expression of (\ref{eqn-3.29}), multiplied by $\varsigma$ or by $\varepsilon$ (of our choice). Since we take 
 $N= (\varepsilon +\varsigma)^{-\frac{1}{2}-\varkappa}$ as dictated by (\ref{eqn-3.28}), it does not exceed the right-hand expression of (\ref{eqn-3.31}).
 
 Unfortunately, we are not that lucky with $n\ge N$: we have only estimates 
 \begin{gather}
 |\sum_{k} \bigl(\lambda_{n,k}(\varepsilon ,\varsigma) - \lambda_{n,k}(\varepsilon,0)\bigr)|\le 
 C r\bigl(\varsigma \|\langle x\rangle^{-3/2}U  \|_{\sL^1}+  \varepsilon r^{3/2}\bigr)n^{-3},
 \label{eqn-3.33}\\
 \shortintertext{and}
 |\sum_{k} \bigl(\lambda_{n,k}(\varepsilon ,\varsigma) - \lambda_{n,k}(0,\varsigma)\bigr)|\le 
 C r\bigl(\varsigma \|\langle x\rangle^{-3/2}U  \|_{\sL^1}+  \varepsilon r^{3/2}\bigr)n^{-3}
 \label{eqn-3.34}\\
 \intertext{which follow from estimate (\ref{eqn-2.11}): namely, we have the same estimate for }
 |\sum_{k} \bigl(\lambda_{n,k}(\varepsilon ,\varsigma) - \lambda_{n,k}^0\bigr)|.
 \notag
 \end{gather}
Then the sum over $n\colon n\ge N= (\varepsilon +\varsigma)^{-\frac{1}{2}-\varkappa}$ does not exceed the right-hand expression of (\ref{eqn-3.31}).
\end{proof}

\chapter{Electronic Density}
\label{sect-4}

\begin{proposition}\label{prop-4.1}
Under assumption \textup{(\ref{eqn-1.7})} the following estimate holds:
\begin{equation}
\varsigma \int U\rho_\Psi \,dx  \le 
\Tr \bigl[(H_{W+\nu}^-\bigr]   -  \Tr \bigl[H_{W+\varsigma U+\nu}^-\bigr] +C Z^{5/3-\delta}
\label{eqn-4.1}
\end{equation}
with $\delta=\delta(\sigma)>0$ as $\sigma>0$ and $\delta=0$ as $\sigma=0$ in 
\textup{(\ref{eqn-1.7})}.
\end{proposition}

\begin{proof}
We know the following upper estimate\footnote{\label{foot-6} See f.e. \cite{monsterbook}, Section~\ref{monsterbook-sect-25-4} and \cite{ivrii:relativistic-1} in non-relativistic and relativistic settings respectively.}
\begin{equation}
\blangle\sfH_{N,Z} \Psi,\,\Psi\brangle \le   \Tr \bigl[H_{W+\nu}^-\bigr] +\nu N - \frac{1}{2} \D(\rho^\TF,\,\rho^\TF)+ \Dirac + CZ^{5/3-\delta}\,,
\label{eqn-4.2}
\end{equation}
where  $\coloneqq W^\TF$ and $\rho^\TF$ are \emph{Thomas-Fermi potential} and
\emph{Thomas-Fermi density}\footnote{\label{foot-7} See f.e.  \cite{monsterbook}, Subsection~\ref{monsterbook-sect-25-1-3}, Volume V.}, $\nu$ is a chemical potential\footnote{\label{foot-8} See f.e.  \cite{monsterbook}, Subsection~\ref{monsterbook-sect-25-2-3}, Volume V. }
 and $\D(f,g)=\iint |x-y|^{-1}f(x)g(y)\,dxdy$.

On the other hand, similarly to the deduction of the lower estimate in \cite{monsterbook}, Subsection~\ref{monsterbook-sect-25-2-1}, for \emph{any potential} $W'=W+U$
\begin{multline}
\blangle\sfH_{N,Z} \Psi,\,\Psi\brangle\\
\begin{aligned}
\ge& \sum _n (H_{V,n} \Psi,\Psi) +\frac{1}{2}\D(\rho_\Psi,\rho_\Psi) -CZ^{5/3-\delta}\\
=&\sum _n (H_{W+U,n} \Psi,\Psi) + \int (W+U-V)\rho_\Psi \,dx +\frac{1}{2}\D(\rho_\Psi,\rho_\Psi) -CZ^{5/3-\delta}\\
=&\sum _n (H_{W+U,n} \Psi,\Psi) + \int U\rho_\Psi \,dx + \frac{1}{2}\D(\rho_\Psi-\rho^\TF,\rho_\Psi-\rho^\TF )\\
&\hskip220pt-\frac{1}{2}\D(\rho^\TF,\,\rho^\TF) -CZ^{5/3-\delta}\\
\end{aligned}\\
\ge \Tr \bigl[ H_{W+U+\nu }^-\bigr] + \nu N  + \int U\rho_\Psi \,dx + \frac{1}{2}\D(\rho_\Psi-\rho^\TF,\rho_\Psi-\rho^\TF )\\
-\frac{1}{2}\D(\rho^\TF,\,\rho^\TF) -CZ^{5/3-\delta}.
\label{eqn-4.3}
\end{multline}

Therefore combining (\ref{eqn-4.2}), (\ref{eqn-4.2}) we arrive to (\ref{eqn-4.1}).
\end{proof}

\begin{remark}\label{rem-4.2}
Replacing in (\ref{eqn-4.1}) $U$ by $-U$ and then multiplying by $-1$, we  estimate the left-hand expression 
of (\ref{eqn-4.1}) from below
\begin{equation}
\varsigma \int U\rho_\Psi \,dx  \ge 
\Tr \bigl[(H_{W-\varsigma U+\nu}^-\bigr]   -  \Tr \bigl[H_{W+\nu}^-\bigr] -C Z^{5/3-\delta}.
\label{eqn-4.4}
\end{equation}
\end{remark}

In this study we consider $\rho_\Psi$  on rather short distance from the nucleus $\y_m$ and want to replace $W^\TF$ by 
$V^0_m= Z_m |x-\y_m|^{-1}$\,\footnote{\label{foot-9} In  \cite{ivrii:el-den-2} we use microlocal approach to study $\rho_\Psi$ on the larger distances from all nuclei and do not do such replacement.}. 

Let $U$ be supported in $B(\y_m,a)$ with some fixed $m$ and satisfy there $|U|\le Z^2$. Further, let $\phi$ be a smooth $a$-admissible  function\footnote{\label{foot-10} I.e. $|D^\alpha \phi |\le C_\alpha b^{-|\alpha|}$ for all $\alpha$.},   $\phi(x)=1$ for $|x-\y_m|\le b$ and $\phi (x)=0$ for $|x-\y_m|\ge 2b$, $b\ge \min(Z^{-1+\delta'},\, 2a)$ with small $\delta'>0$.

Rescaling $x\mapsto (x-\y_m) Z_m$, $\tau \mapsto \tau Z_m^{-2}$ we find ourselves in the framework of Sections~\ref{sect-2} and~\ref{sect-3}. On the other hand,  due to \cite{monsterbook}, Sections~\ref{monsterbook-sect-25-4} and~\ref{monsterbook-sect-12-5} (Volume II)
\begin{equation}
|\Tr \bigl[[(1-\phi ) H_{W+\nu}^-\bigr]   -  \Tr \bigl[(1-\phi) H_{W+\varsigma U+\nu}^-\bigr]|\le C Z^{5/3-\delta}
\label{eqn-4.5}
\end{equation}
(with $\delta=\delta(\sigma)$) and we can insert $\phi$ into right-hand expression of (\ref{eqn-4.1}): 
\begin{equation}
\varsigma \int U\rho_\Psi \,dx  \le 
\Tr \bigl[(\phi H_{W+\nu}^-\bigr]   -  \Tr \bigl[\phi H_{W+\varsigma U+\nu}^-\bigr] +C Z^{5/3-\delta}.
\label{eqn-4.6}
\end{equation}
Indeed, $W+\varsigma U=W'$ is smooth $b$-admissible function  in the zone\linebreak $\{x\colon b \le |x-\y_m|\le 9b\}$.

Then since $|V^0_m-W|\le CZ^{4/3}$ in $2b$-vicinity of $\y_m$ and $|\nu|\le CZ^{4/3}$  we can replace 
$W+\nu$ by $V^0_m$  with an error not exceeding $CZ^{4/3} \times r^{3/2}$, where we prefer to use a scaled distance $r=Z_m a$ rather than $a$: 
\begin{equation}
\varsigma \int U\rho_\Psi \,dx  \le 
\Tr \bigl[(\phi H_{V ^0_m}^-\bigr]   -  \Tr \bigl[\phi H_{V^0_m+\varsigma U }^-\bigr] +C Z^{5/3-\delta}+ CZ^{4/3} r^{3/2}.
\label{eqn-4.7}
\end{equation}
Indeed, using arguments of \cite{monsterbook}, Sections~\ref{monsterbook-sect-25-4} and~\ref{monsterbook-sect-12-5} (Volume II), we can first replace $W$ by $W_m$, coinciding with $V^0_m$ on $B(\y_m ,3b)$ and with $W$ outside of $B(\y_m ,4b)$\,\footnote{\label{foot-11} This replacement brings an extra error--the last term in the right-hand expression of (\ref{eqn-4.7}). There should be $b$ rather than $a$ but $b\asymp a$ unless $b\asymp Z^{-1+\delta''}$ in which case the previous term is larger.} and then replace $W_m$ by $V^0_m$ using the same arguments and that $\supp (\phi)$ is ``smaller''. 

Note that the last term in the right-hand expression of (\ref{eqn-4.7}) is smaller than the previous one, and we can drop it, as 
$r\le Z^{2/9-2\delta/3}$, or 
\begin{gather}
a\le Z^{-7/9-2\delta/3}.
\label{eqn-4.8}
\end{gather}

Next, using rescaling $x\mapsto (x-\y_m)Z_m$, $U\mapsto \bar{U}= Z_m^{-2}U$ we find ourselves in the framework of Sections~\ref{sect-2} and~\ref{sect-3}  with $r= Z_m a$. Therefore, using decomposition of Remark~\ref{rem-3.5}\ref{rem-3.5-i} and Proposition~\ref{prop-3.6}, and Remark~\ref{rem-4.2},  we arrive to
\begin{gather}
|\int \bar{U}\bigl(\rho_\Psi -\rho_m)\,dx | 
\le  C\varsigma^{1-\varkappa}r \|\langle x\rangle^{-3/2} \bar{U}\|^*_{\sL^1} +C\varsigma^{-1} Z^{-1/3-\delta}
\label{eqn-4.9}
\end{gather}
with an arbitrarily small exponent $\varkappa>0$ and $\bar{U}=Z^{-2}U$, satisfies  (\ref{eqn-1.11}). Here and below 
$\| . \|^*_{\sL^1}$ is calculated \emph{after rescaling}; in comparison with $\| . \|_{\sL^1}$, calculated \emph{before rescaling}, it has an extra factor $Z^3$.

However we can do better than this.  Namely, observe that $W'_m\coloneqq W^\TF-V^0_m$ satisfies
\begin{equation}
|W'_m(x)- W'_m(\y_m)|\le CZ_m ^{3/2}|x-\y_m|^{1/2}
 \label{eqn-4.10}
\end{equation}
\emph{before rescaling}\footnote{\label{foot-12} It follows from $|\Delta W' |\le C Z_m^{3/2}|x-\y_m|^{-3/2}$ due to Thomas-Fermi equation.}
and therefore, if we replace $V_m^0$ by $V_m^0 + W'_m(0)-\nu $ in (\ref{eqn-4.7}), the last term in the right-hand expression is replaced by 
$C Z r^2$:
\begin{multline}
\varsigma \int U\rho_\Psi \,dx  \le 
\Tr \bigl[(\phi H_{V ^0_m+W'_m(0)-\nu }^-\bigr]   -  \Tr \bigl[\phi H_{V^0_m+W'_m(0)-\nu +\varsigma U }^-\bigr] \\
+C Z^{5/3-\delta}+ CZ  r^2.
\label{eqn-4.11}
\end{multline}
The last term in the right-hand expression of (\ref{eqn-4.11}) is smaller than the previous one, and we can drop it  as 
$r\le Z^{1/3-\delta/2}$, or 
\begin{equation}
a\le Z^{-2/3-\delta/2}.
\label{eqn-4.12}
\end{equation}

\begin{remark}\label{rem-4.3}
Therefore,   (\ref{eqn-4.9}) remains true under this assumption.  Indeed, adding a constant to $V^0$ we just shift all eigenvalues of $V^0$ and $V^0+\varsigma U$ and preserve all projectors $\uppi^0_n$. 
\end{remark}

To improve our results on the  distances $a\ge Z^{-2/3-\delta/2}$ we  will use Proposition~\ref{prop-3.8}. Due to (\ref{eqn-4.10}) 
$\varepsilon = Z^{-1/2}a^{1/2}$ and to satisfy $\varepsilon r^{3/2}\le Z^{-\varkappa}$ we  assume that
\begin{equation}
a \le Z^{-1/2-\varkappa}.
\label{eqn-4.13}
\end{equation}

Then due to Proposition~\ref{prop-3.8}  the right-hand expression of (\ref{eqn-4.9}) should be replaced by 
\begin{multline*}
C\varsigma^{-1} (\varsigma +\varepsilon)  \bigl(\varsigma r\|\langle x\rangle^{-3/2}\bar{U} \|^*_{\sL^1} + \varepsilon r ^{5/2}\bigr) Z^{\varkappa}
+ C\varsigma^{-1}Z^{-1/3-\delta}\\[3pt]
\asymp 
C\varsigma r\|\langle x\rangle^{-3/2}\bar{U} \|^*_{\sL^1} Z^{\varkappa} + 
C\varepsilon r^{5/2}Z^{\varkappa} +
C\varsigma^{-1}\bigl(\varepsilon^2 r^{5/2}Z^{\varkappa} +Z^{-1/3-\delta}\bigr);
\end{multline*}
as we will see $\varsigma \ge Z^{-s}$ with sufficiently large exponent $s$. We rewrite
the last expression as
\begin{multline}
C\Bigl(\varsigma Z^4a \|\langle Z(x-\y_m)\rangle^{-3/2}\bar{U} \|_{\sL^1} \\
+  Z^{2}a^{3} +
\varsigma^{-1}\bigl(Z^{3/2}a^{7/2}  +Z^{-1/3-\delta}\bigr)\Bigr)Z^\varkappa,
\label{eqn-4.14}
\end{multline}
which would lead to the same final estimate (\ref{eqn-1.12})--(\ref{eqn-1.14}).

Minimizing (\ref{eqn-4.14}) by $\varsigma\colon \varsigma \le  Z^{-\varkappa}r^{-3/2}=Z^{-3/2-\varkappa}a^{-3/2}$, we arrive to 
\begin{multline}
C\Bigl( \bigl(Z^4a \|\langle Z(x-\y_m)\rangle^{-3/2}\bar{U}\|_{\sL^1} \bigr)^{1/2} \bigl( Z^{3/2} a^{7/2}  + Z^{-1/3-\delta}\bigr)^{1/2}\\
+ Z^2a^3 + Z^{3/2}a^{3/2} \bigl( Z^{3/2} a^{7/2} + Z^{-1/3-\delta}\bigr)\Bigr) Z^\varkappa.
\label{eqn-4.15}
\end{multline}
Therefore writing  $U$ instead of $\bar{U}$, we have
\begin{multline}
|\int U\bigl(\rho_\Psi -\rho_m)\,dx | \le\\
C\Bigl( \bigl(Z^{3/2}a ^{3/2}\|\langle Z(x-\y_m)\rangle^{-3/2} U \|_{\sL^1} \bigr)^{1/2} \bigl( Z^{4} a^{3}  + Z^{13/6-\delta}a^{-1/2}\bigr)^{1/2}\\
+ Z^2a^3 +   Z^{3} a^{5} + Z^{7/6-\delta}a^{3/2}\bigr)\Bigr) Z^\varkappa.
\label{eqn-4.16}
\end{multline}
Since $Z^3a^5\le Z^2a^3$ we arrive to  estimate (\ref{eqn-1.12})--(\ref{eqn-1.14}), which concludes the proof of Theorem~\ref{thm-1.3}.

Observe that 
$\supp (U)\subset B(\y_m, a)\setminus B(\y_m, a/2)$ implies 
\begin{gather*}
Z^{3/2}a ^{3/2}\|\langle Z(x-\y_m)\rangle^{-3/2} U \|_{\sL^1} \asymp
\| U \|_{\sL^1}.
\end{gather*}

To prove  Corollary~\ref{cor-1.4}\ref{cor-1.4-i} we plug 
\begin{equation}
U(x)= \upchi_X(x) \sign \bigl(\rho_\Psi(x)  -\rho_m(x)\bigr)
\label{eqn-4.17}
\end{equation}
with characteristic function $\upchi_X$ of $X$; then the left-hand expression becomes $\|\rho_\Psi - \rho_m\|_{\sL^1(X)}$ and
$\|U\|_{\sL^1}= \mes(X)$.

To prove  Corollary~\ref{cor-1.4}\ref{cor-1.4-ii} we plug 
\begin{gather}
U(x)= \epsilon _0 \omega ^{1-p}\upchi_X(x) \sign \bigl(\rho_\Psi(x)  -\rho_m(x)\bigr)| \rho_\Psi(x)  -\rho_m(x)|^{p-1}\,,
\label{eqn-4.18}\\
\intertext{which satisfies (\ref{eqn-1.11}) due to assumptions $\rho_\Psi \le \omega $ in $X$ and estimate $\rho_m \le CZ^{3/2}a^{-3/2}\le C\omega  $  and since $\|U\|_{\sL^1}=\omega^{-1}J_{p-2}$, we get that}
J_{p}\le C \bigl((\omega^{-1} F)^{1/2} J_{p-1}^{1/2} +  G\bigr)
\intertext{with   $J_p\coloneqq \omega ^{1-p} \int_X |\rho_\Psi -\rho_m|^p\,dx$ and $J_0= \omega  \mes(X) $, and therefore}
J_{p}\le C\bigl( (\omega^{-1} F)^{1-2^{-p}} J_0^{2^{-p}} + (\omega^{-1}F)^{1-2^{1-p}} G^{2^{1-p}} +G\Bigr),
\notag
\end{gather}
which implies (\ref{eqn-1.18}).

\begin{appendices}
\chapter{Properties of eigenfunctions of Coulomb-Schr\"odinger operator}
\label{sect-A}

\section{General}
\label{sect-A.1}

Consider operator $H^0=-\Delta -r^{-1}$; $r=|x|^{-1}$. It is known that its eigenfunctions, corresponding to eigenvalues $-\frac{1}{4n^2}$, with norm $1$,  are
(in the spherical coordinates)
\begin{equation}
u_{n,l,m}(r,\varphi,\theta)=R_{n,l}(r)Y_{l,m}(\varphi,\theta) 
\label{eqn-A.1}
\end{equation}
where $Y_{l,m}(\varphi,\theta)$ are spherical functions with $m=-l,-l+1,\ldots, l-1,l$, $l=0,1,\ldots, n$ and 
\begin{equation}
R_{n,l}(r)= \sqrt{\frac{(n-l-1)!}{2n^4 (n+l)!}}  \bigl( \frac{r}{n} \bigr) ^le^{-r/2n}
L^{(2l+1)}_{n-l-1} \bigl(\frac{r}{n}\bigr)
\label{eqn-A.2}
\end{equation}
with associated Laguerre polynomials
\begin{align}
L_n^{(k)}(z) &= \frac{1}{n!}z^{-k} e^z \frac{d^n}{dz^n} \bigl(e^{-z} z^{n+k}\bigr)
         =\sum_{0\le j\le n} \frac{(n+k)!}{(n-j)!(k+j)!j! } (-z)^{j}.  \label{eqn-A.3} 
\end{align}
Then 
\begin{align}
L_{n-l-1}^{(2l+1)}(z) =\sum_{0\le j\le n-l-1} \frac{(n+l)!}{(n-l-1-j)! (2l+1+j)!j!} (-z)^j.
\label{eqn-A.4}
\end{align}

Then $v\coloneqq v_{n,l}=R_{n,l}(r) r$ satisfies 
\begin{equation}
-v '' + \frac{l(l+1)}{r^2}v -\frac{1}{r} v=\lambda_n v
\label{eqn-A.5}
\end{equation}
and is $(n-l)$-th eigenfunction and $\lambda_n$ is $(n-l)$-th eigenvalue of such operator na $\sL^2(\bR^+)$ and the associated variational form is
\begin{equation}
\int \bigl( v^2+ \frac{l(l+1)}{r^2}v^2 -\frac{1}{r} v^2\bigr)\,dr. 
\label{eqn-A.6}
\end{equation}

\section{Zeroes}
\label{sect-A.2}

\begin{proposition}
\label{prop-A-1}
$v_{n,l}$ has exactly $(n-l-1)$ zeroes
\begin{gather}
l(l+1) < r_* < r_1 < \ldots < r_{n-l-1} < r^*< 4n^2,
\label{eqn-A.7}\\
\intertext{where $r_*<r^*$ are two roots of}
W(r)\coloneqq \frac{1}{r}-\frac{l(l+1)}{r^2} =-\lambda_n 
\label{eqn-A.8}\\
\shortintertext{and}
r_{k} \asymp (l+k+1)^2.
\label{eqn-A.9}
\end{gather}
\end{proposition}

\begin{proof}
Standard variational methods imply that $v_{n,l}(r)$ has exactly $(n-l-1)$ zeroes.

Further, equation (\ref{eqn-A.5}) and $v(0)=v(\infty)=0$ imply that all zeroes are simple and satisfy 
$W(r)> -\lambda_n$, which implies (\ref{eqn-A.7}).

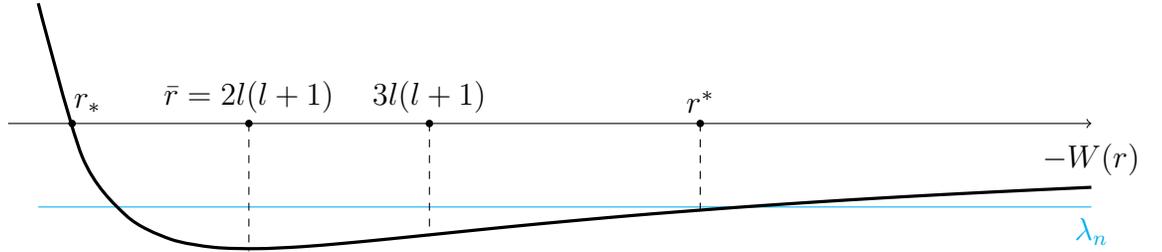
\begin{figure}[h]
\centering
\begin{tikzpicture}
\begin{scope}[xscale=.4, yscale=40]
\draw[thin,->] (4,0)--(40,0);
\draw[thin,cyan] (5,-1/36)--(40,-1/36) node[below] {$\lambda_n$};
\draw[very thick, domain=5:40, smooth] plot (\x, {6/\x^2-1/\x}) node[above] {$-W(r)$};
\end{scope}
\fill (2.45,0) circle (.05)  node[above] {$\quad r_*$};
\fill (4.8,0) circle (.05)  node[above] {$\bar{r}=2l(l+1)$};
\draw[thin, dashed] (4.8,0)--(4.8,-1.7);
\fill (7.2,0) circle (.05)  node[above] {$3l(l+1)$};
\draw[thin, dashed] (7.2,0)--(7.2,-1.5);
\fill (10.8,0) circle (.05)  node[above] {$r^*$};
\draw[thin, dashed] (10.8,0)--(10.8,-1.2);
\end{tikzpicture}
\caption{Marked points: left-bound for $\lambda_{l+1}$, minimum point, inflection point, right bound for $\lambda_{n}$. }
\end{figure}

Consider two points $r_*< r' < r''< r^*$ and observe that $W(r)\ge \min (W(r'),W(r''))$ on $(r',r'')$. Then $v$ must have a zero on $(r',r'')$ as long as $r''-r' > \pi /\sqrt{\min  (W(r'),W(r''))+\lambda_n}$. 

Therefore, if $cl(l+1) \le r\le c^{-1}n^2$, then in $2\pi r^{1/2}$-vicinity of $r$ must be zero of $v$. Thus the distance between two consecutive zeroes $r_k$ and $r_{k\pm 1}$ with  $cl(l+1)\le r_k \le  c^{-1}n^2$ is $O(r_k^{1/2})$. From this one can prove easily that then
\begin{gather*}
|r_k-r_{k\pm 1}|= \pi r_k^{1/2}\Bigl(1+ O\bigl(\frac{l(l+1)}{r_k}  +  \frac{r_k}{n^2}\bigr) \Bigr)
\intertext{and therefore}
|r_k^{1/2}-r_{k\pm 1}^{1/2}|= \frac{\pi}{2} \Bigl(1+ O\bigl(\frac{l(l+1)}{r_k}  +  \frac{r_k}{n^2}\bigr) \Bigr).
\end{gather*}
Further, there must be zeroes to the left of $cl(l+1)$ and to the right of $c^{-1}n^2$. 

Finally, observe that 
$W(r)\ge 1/4l(l+1)$ and therefor e $(r_{k+1}-r_k)\ge 2\pi \sqrt{l(l+1)}$ and therefore $(r_{k+1}^{1/2}-r_k^{1/2})\ge c^{-1}\pi$ for all $k$ with $r_k \le cl(l+1)$. This implies (\ref{eqn-A.9}). 
\end{proof}

\begin{remark}\label{rem-A.2}
Since $-\partial^2\ge \frac{1}{4r^2}$ in fact $r_k \ge (l+1/2)^2$; in particular, $r_k\ge \frac{1}{4}$ as $l=0$.
\end{remark}

Let us analyze $r_k$ more carefully. Due to monotonicity of $W(r)$ on $(0,\bar{r})$ and $(\bar{r},\infty)$ we conclude that 
\begin{align}
&r_k \ge \bar{r}\implies s_{k+1}\ge s_k, 
&& \frac{\pi}{\sqrt{W(r_k)+\lambda_n}}\le  s_k \le \frac{\pi}{\sqrt{W(r_{k+1})+\lambda_n}}
\label{eqn-A.10}\\
\shortintertext{and}
&r_k \le \bar{r}\implies s_{k-1}\ge s_k, 
&& \frac{\pi}{\sqrt{W(r_k)+\lambda_n}}\le  s_k \le \frac{\pi}{\sqrt{W(r_{k-1})+\lambda_n}}.
\label{eqn-A.11}
\end{align}

Consider $r_k$ close to $r^*$. In this case $W(r)\asymp (r^*-r)/r^{*\,2}$ and $s_k \asymp r^* /\sqrt{(r^*-r_k)}$ provided $W(r_{k+1}\asymp W(r_k)$; then $(r^*-r_k) \ge Cs_k$ as\linebreak  $(r^*-r_k)\ge r^*{2/3}$ and we arrive to Statement~\ref{prop-A.3-i} below;  Statement~\ref{prop-A.3-ii} is proven the same way\footnote{\label{foot-13} Assumption  $n-l\ge 3$ is needed to have at least $2$ zeroes.}. From the same arguments follows Statement~\ref{prop-A.3-iii}.

\begin{proposition}\label{prop-A.3}
Let $n-l \ge 3$, $n\ge c$. Then
\begin{enumerate}[label=(\roman*), wide, labelindent=0pt]
\item\label{prop-A.3-i} 
As $r_k\asymp r^*$ and $(r^*-r_k)\ge Cn^{4/3}$
\begin{equation}
s_k \asymp  r^* /\sqrt{(r^*-r_k)} \qquad\text{and}\qquad r^*-r_k \asymp r^{*\,2/3} (n-l-k)^{2/3}.
\label{eqn-A.12}
\end{equation}

\item\label{prop-A.3-ii} 
As $r_k\asymp r_*$ and $(r_k-r_*)\ge Cl^{4/3}$
\begin{equation}
s_k \asymp  r_* /\sqrt{(r_k-r_*)} \qquad\text{and}\qquad r_k-r_* \asymp r^{*\,2/3} k^{2/3}.
\label{eqn-A.13}
\end{equation}

\item\label{prop-A.3-iii}
There are no more than $C'$ zeroes in the zones $\{r\colon r\le r_*+ C r_*^{2/3}\}$ and
$\{r\colon r\ge r^*- Cr^{*\,2/3}\}$.
\end{enumerate}
\end{proposition}

\section{Estimates}
\label{sect-A.3}

From(\ref{eqn-A.5}) it follows that
\begin{equation*}
\int_{r_k}^{r_{k+1}} (-2v'' v' +\frac{2l(l+1)}{r^2}vv' - \frac{2}{r}vv'-2\lambda_n vv')\,dr =0 
\end{equation*}
and then
\begin{equation}
-v'^2 (r_{k+1})+v'^2(r_k)  =  
\int_{r_k}^{r_{k+1}}\bigl( -\frac{2l(l+1)}{r^3}+ \frac{1}{r^2}\bigr)v^2\,dr,
\label{eqn-A.14}
\end{equation}
where $k=0,\ldots ,n-l-1$, and $r_0\coloneqq 0$, $r_{n-l}\coloneqq \infty$.

Then we conclude that
\begin{align}
&|v'(r_k)|> |v'(r_{k+1})|  &&\text{for \ } r_k \ge 2l(l+1)
\label{eqn-A.15}\\
\shortintertext{and}
&|v'(r_k)|< |v'(r_{k+1})| &&\text{for \ } r_{k+1} \le 2l(l+1).
\label{eqn-A.16}
\end{align}

Consider first $r_k\colon (1+\epsilon)r_*\le r_k\le (1-\epsilon)r^*$\,\footnote{\label{foot-14} Which is possible if and only if $l\le (1-\epsilon')n$.}. 

Then one can see easily that 
\begin{multline}
v  (r)= v'_k(r_k) \frac{s_k}{\pi}\sin \bigl(\frac{\pi (r-r_k)}{s_k}\bigr)\bigl(1+O(r_k^{-1/2})\bigr)\\
\text{for\ \ } r_k\le r\le r_{k+1}
\label{eqn-A.17}
\end{multline}
and we calculate the right-hand expression of (\ref{eqn-A.14}) arriving to
\begin{equation}
v'^2 (r_{k+1})=v'^2(r_k) \Bigl[1- 
\frac{1}{2} r_k^{-2}\frac{s_k^3}{\pi^2}\Bigl(1-\frac{2l(l+1)}{r_k}+  O(r_k^{-1/2})\Bigr)\Bigr].
\label{eqn-A.18}
\end{equation}
Then in virtue of (\ref{eqn-A.11})
\begin{multline}
v'^2 (r_{k+1})=\\
v'^2(r_k) \Bigl[1- 
\frac{1}{2} r_k^{-1}s_k\Bigl(1-\frac{2l(l+1)}{r_k})\Bigr)\Bigl(1 -\frac{l(l+1)}{r_k} +\lambda_n   \Bigr)^{-1}+
O(r_k^{-1})\Bigr].
\label{eqn-A.19}
\end{multline}
Then
\begin{multline}
v'^2 (r_{k+1})r_{k+1}^{1/2}=\\
v'^2(r_k)r_k^{1/2} \Bigl[1+ \frac{1}{2} r_k^{-1}s_k
\Bigl(\frac{l(l+1)}{r_k}+\lambda_n)\Bigr)\Bigl(1 -\frac{l(l+1)}{r_k} +\lambda_n   \Bigr)^{-1}
+ O(r_k^{-1})\Bigr]=\\
v'^2(r_k)r_k^{1/2}  (1+ \varepsilon_k)
\label{eqn-A.20}
\end{multline}
with $\varepsilon_k =O\bigl(\dfrac{l(l+1)}{r_k^{3/2}}+\dfrac {1}{r_k}\bigr)$. In virtue of (\ref{eqn-A.9}) 
$\sum \varepsilon _k \le C$ and therefore we arrive to 
\begin{equation}
v'^2 (r_{k})r_{k}^{1/2}\asymp v'^2 (r_{m})r_{m}^{1/2} .
\label{eqn-A.21}
\end{equation}

Since 
\begin{gather}
\max_{ r_k\le r\le r_{k+1}}  |v(r)| \asymp |v'(r_k) (r_{k+1}-r_k)
\label{eqn-A.22}\\
\shortintertext{we arrive to}
r_k ^{-1/4} \max_{ r_k\le r\le r_{k+1}}  |v(r)|  
\asymp r_m ^{-1/4} \max_{ r_m\le r\le r_{m+1}}  |v(r)|.
\label{eqn-A.23}
\end{gather}

Then  taking $m\asymp n$ and using (\ref{eqn-A.9}) we conclude that $p_k\asymp p_m k^{1/2}n^{-1/2}$ with $p_k$ the left-hand expression of  (\ref{eqn-A.21}). Since
\begin{equation}
\int _{r_k}^{r_{k+1}}v^2(r)\,dr \asymp  \max_{ r_k\le r\le r_{k+1}}  v^2(r) (r_{k+1}-r_k)
\label{eqn-A.24}
\end{equation}
we conclude that it is $\asymp p_n^2 k^2 n^{-1}$ and therefore their sum is $\asymp p_m^2 n^2$; since it should not exceed $1$, we conclude that $p_m \le Cn^{-1}$; applying $p_k\asymp p_m k^{1/2}n^{-1/2}$ again we arrive to

\begin{proposition}\label{prop-A.4}
Assume that
\begin{gather}
l\le (1-\epsilon')n.
\label{eqn-A.25}
\intertext{Then for $(1+\epsilon)r_*\le r_k \le  (1-\epsilon)r^*$} 
\max_{ r_k\le r\le r_{k+1}}  |v(r)| \le C' r^{1/4}n^{-3/2}.
\label{eqn-A.26}
\end{gather}
\end{proposition}

\begin{remark}\label{rem-A.5}
It follows from the arguments below that actually in (\ref{eqn-A.26}) there is ``$\asymp$'' sign.
\end{remark}

Consider now zone $\{r\colon \max ((1-\epsilon) r^*\le r \le r^*-Cr^{*\, 2/3}\}$, again, under assumption (\ref{eqn-A.26}).

Recall, that in  this zone \ref{eqn-A.12}) holds  and also
\begin{gather}
W(r)= (r^*-r)r^{*\,-2} \Bigl( 1 +O\bigl(\frac{r^*}{(r^*-r_k)^{3/2}}\bigr)\Bigr)
\label{eqn-A.27}\\
\shortintertext{and}
s_k = \frac{\pi r^* }{\sqrt{r^*-r_k}}\Bigl( 1 +O\bigl(\frac{r^*}{(r^*-r_k)^{3/2}}\bigr)\Bigr).
\label{eqn-A.28}\\
\shortintertext{Then for $r_k\le r\le r_{k+1}$}
v(r) = v'(r_k) \frac{s_k}{\pi} \sin \bigl(\frac{\pi (r-r_k)}{s_k}\bigr)\Bigl( 1 +O\bigl(\frac{r^*}{(r^*-r_k)^{3/2}}\bigr)\Bigr).
\label{eqn-A.29}
\end{gather}
Then (\ref{eqn-A.14}) implies that
\begin{equation*}
v'^2(r_{k+1}) -v'^2(r_k) = -\frac{1}{2} v'^2 (r_k) \frac{s_k^3}{\pi^2 r^{*\,2}} \Bigl( 1- \frac{2l(l+1)}{r^{*}}\Bigr) \Bigl( 1 +O\bigl(\frac{r^*}{(r^*-r_k)^{3/2}}\bigr)\Bigr),
\end{equation*}
where (\ref{eqn-A.25}) ensures that the first large parentheses are disjoint $0$.  Then plugging (\ref{eqn-A.28}) we conclude that
\begin{multline}
v'^2(r_{k+1}) =\\
v'^2(r_k)\Bigr[ 1 -\frac{1}{2}  \frac{s_k}{(r^*-r_k)} \Bigl( 1- \frac{2l(l+1)}{r^{*}}\Bigr) \Bigl( 1 +O\bigl(\frac{r^*}{(r^*-r_k)^{3/2}}\bigr)\Bigr)\Bigr].
\label{eqn-A.30}
\end{multline}
Then
\begin{multline}
v'^2(r_{k+1})(r^*-r_{k+1})^{-1/2} =\\
v'^2(r_k)(r^*-r_{k+1})^{-1/2}\Bigr[ 1 +  \frac{(r_{k+1} -r_{k})}{(r^*-r_k)}\frac{l(l+1)}{r^{*}} +\varepsilon_k\Bigr]
\label{eqn-A.31}
\end{multline}
with $\sum_k\varepsilon_k <\infty$.

Therefore
\begin{gather*}
v'^2(r_{k})(r^*-r_{k})^{-1/2+\sigma}\asymp v'^2(r_{m})(r^*-r_{m})^{-1/2+\sigma},
\shortintertext{and}
|v'(r_{k})| \asymp |v'(r_{m})| (r^*-r_{k})^{1/4-\sigma/} (r^*-r_{m})^{-1/4+\sigma/2},\\
\shortintertext{and finally}
\max_{ r_k<r<r_{k+1}} |v (r)| \asymp \max_{r_m<r<r_{m+1}} |v (r)| (r^*-r_{k})^{-1/4-\sigma/2} (r^*-r_{m})^{1/4+\sigma/2}.
\end{gather*}

Taking $m\asymp n$ such that $r_m\le (1-\epsilon) r^*$ we arrive in virtue of Proposition~\ref{prop-A.4} to

\begin{proposition}\label{prop-A.6}
Under assumption \textup{(\ref{eqn-A.25})} for $(1-\epsilon)r^*\le r <  r^* - Cr^{*\,2/3}$ 
\begin{equation}
|v (r)| \le C' (r^*-r )^{-1/4-\sigma/2} n^{-1/2+\sigma} ,\qquad \sigma =\frac{l(l+1)}{r^*}.
\label{eqn-A.32}
\end{equation}

\end{proposition}

Consider now zone $\{r\colon  r_*+Cr_*^{2/3}\le r \le (1+\epsilon)r_* \}$ provided (\ref{eqn-A.25}). 
The same arguments lead us to
\begin{equation*}
\max_{ r_k\le r\le r_{k+1}}  |v(r)|  \asymp 
(r_k-r_*) ^{-3/4+\sigma'/2} (r_m-r_*) ^{3/4-\sigma'/2} 
\max_{ r_m\le r\le r_{m+1}}  |v(r)|.
\end{equation*}
and taking $m\asymp l$ such that $r_m\le (1+\epsilon)r_*$ and therefore $r_m \lesssim n^{-3/2}l^{1/2}$, we arrive to

\begin{proposition}\label{prop-A.7}
Under assumption \textup{(\ref{eqn-A.25})} for $r_*+Cr_*^{2/3} \le r \le  (1+\epsilon)r_*$ 
\begin{equation}
|v (r)| \le C'   (r -r_*) ^{-3/4+\sigma'/2}  n^{-3/2} l ^{2-\sigma'}   ,\qquad \sigma' =\frac{l(l+1)}{r_*}.
\label{eqn-A.33}
\end{equation}
\end{proposition}

\begin{proposition}\label{prop-A.8}
Let  assumption \textup{(\ref{eqn-A.25})} be fulfilled. Then
\begin{enumerate}[label=(\roman*), wide, labelindent=0pt]
\item\label{prop-A.8-i}
The following estimates hold
\begin{align}
&|v(r)|\le C'   n^{-5/6-\sigma/3} &&\text{for}\ \ r\ge r^*- Cr^{*\,2/3}
\label{eqn-A.34}\\
\shortintertext{and}
&|v (r)| \le C'   n^{-3/2} l^{1-2\sigma'/3}  &&\text{for}\ \ r\le r_* + Cr_*{2/3}.
\label{eqn-A.35}
\end{align}

\item\label{prop-A.8-ii}
Furthermore, let  $ b\coloneqq C_s r^{*\,2/3}$. Then  
\begin{equation}
|v(r)|\le C' n^{-5/6-\sigma/3} \bigl(\frac{b}{r-r^*}\bigr)^s\qquad \text{for}\ \  r\ge r^*+b.
\label{eqn-A.36}
\end{equation}

\item\label{prop-A.8-iii}
On the other hand, let  $b\coloneqq C_s r_*^{2/3}$. Then  
\begin{equation}
|v(r)|\le C'  n^{-3/2} \bar{l}^{1-\sigma' /3} \bigl(\frac{b}{r_*-r}\bigr)^s\qquad \text{for}\ \  r\le r_*-b.
\label{eqn-A.37}
\end{equation}
\end{enumerate}
\end{proposition}

\begin{proof}
\begin{enumerate}[label=(\roman*), wide, labelindent=0pt]
\item\label{pf-A.8-i}
Estimates (\ref{eqn-A.34}) and  (\ref{eqn-A.35}) for $r^*- Cr^{*\,2/3}\le r\le r^*+ Cr^{*\,2/3}$ and
$r_*- Cr_*^{2/3}\le r\le r_*+ Cr_*^{*2/3}$ follow from  estimates (\ref{eqn-A.32}) and (\ref{eqn-A.33}) and equation (\ref{eqn-A.5}).

\item\label{pf-A.8-ii}
Consider $\phi\in \sC^\infty$, $\phi =0$ on $(-\infty,\frac{1}{2})$ and $\phi =1$ on $(1,\infty)$ and 
$\varphi (r)=\phi ((r-r^*)/a)$. Then multiplying (\ref{eqn-A.5}) by $\varphi v$ and integrating by parts we get
\begin{gather}
\int \varphi (r)v'^2(r)\,dr + \int \bigl( W(r)-\lambda_n  \bigr)\varphi (r)v^2(r)\,dr =
\frac{1}{2}\int \varphi''(r)v^2(r)\,dr 
\notag\\
\intertext{and therefore}
\int_{r^*+a}^\infty   r^{-2} v^2(r)\,dr 
\le Ca^{-3} \int_{r^*+a/2}^{r^*+a}   v^2(r)\,dr
\notag\\
\intertext{which implies after iterations estimate }
\int _{r^*+a}^\infty v^2(r)\,dr \le C' a M^2 \bigl(\frac{b}{a}\bigr)^s, 
\label{eqn-A.38}
\end{gather}
for integral from $r^*+a$ to $r^*+2a$, where $M$ is the right-hand expression of (\ref{eqn-A.34}), which, in turn, implies  
(\ref{eqn-A.38}) in full measure.

Then the same proof implies that 
\begin{equation*}
\int _{r^*+a}^\infty v'^2(r)\,dr\le C'a^{-1}   M^2\bigl(\frac{b}{a}\bigr)^s
\end{equation*}
which combined with (\ref{eqn-A.38}) implies (\ref{eqn-A.36}).

\item\label{pf-A.8-iii}
Statement \ref{prop-A.8-iii} is proven in the same way.
\end{enumerate}\vskip-\baselineskip
\end{proof}

Consider now the case $l\ge (1-\epsilon n$. In this case both $r_*\approx r^*\approx 4n^2\approx r_*/4n^2$ and, $(r^*-r_*)\approx 4n^{3/2}\sqrt{2(n-l)}$,  $W(r)\approx (r-r_*)(r^*-r)r^{*\,-2}$, where $\approx$ means that the ration is close to $1$.

Further $C r^{*,2/3}$ should be  replaced by $Cr^*/(r^*-r_*)^{1/3}$ and we want $Cr^*/(r^*-r_*)^{1/3}\le (r^*-r_*)$ i.e. 
$(r^*-r_*) \ge Cr^{*\,3/4}$ which is equivalent
\begin{equation}
C_0\le (n-l)\le \epsilon n.
\label{eqn-A.39}
\end{equation}

\begin{proposition}\label{prop-A.9}
\begin{enumerate}[label=(\roman*), wide, labelindent=0pt]
\item\label{prop-A.9-i}
Let condition \textup{(\ref{eqn-A.39})} be fulfilled. Then 
\begin{multline}
|v(r)|\le \\ C'  
\left\{\begin{aligned}
& L ^{-1/4+\sigma/2}(r^*-r)^{-1/4-\sigma/2} &&  r^* -\epsilon L \le r\le r^* - Cr^*L ^{-1/3},\\
& L^{-1/2}                   && r_*+\epsilon L \le r\le r^*-\epsilon L,   \\
&  L^{1/4-\sigma'/2}  (r-r_*)^{-3/4+\sigma'/2} \qquad &&  r_* + Cr_* L ^{-1/3} \le r \le  r_* +\epsilon L
\end{aligned}\right.
\label{eqn-A.40}
\end{multline}
with $L=r^*-r_*$.

\item\label{prop-A.9-ii}
Further, 
\begin{equation}
|v(r)|\le C'   L^{-1/3+2\sigma/3}r^{*\,-1/4-\sigma/2} \left\{\begin{aligned}
&1  &&r\ge r^* - Cb,\\
&\bigl(\frac{b}{r-r^*}\bigr)^s \qquad && r\ge r^*+Cb
\end{aligned}\right.
\label{eqn-A.41}
\end{equation}
with $b=r^*L ^{-1/3}$.

\item\label{prop-A.9-iii}
Furthermore,
\begin{equation}
|v(r)|\le C'   L^{1/2-2\sigma'/3}r_*^{-3/4+\sigma'}\left\{\begin{aligned}
&1  &&r\le r^* + Cb,\\
&\bigl(\frac{b}{r_*-r}\bigr)^s \qquad&& r\le r_*-Cb
\end{aligned}\right.
\label{eqn-A.42}
\end{equation}
with $b=r_*L ^{-1/3}$.
\end{enumerate}
\end{proposition}

\begin{proof}
Statement~\ref{prop-A.9-i} is proven in the same way as Propositions~\ref{prop-A.4},~\ref{prop-A.6} and~\ref{prop-A.7}.
Statements~\ref{prop-A.9-ii} and~\ref{prop-A.9-iii} are proven in the same way as Proposition~\ref{prop-A.8}.
\end{proof}

Finally, consider the remaining case $1\le n-l\le C$. In the same way 

\begin{proposition}\label{prop-A.10}
Let $1\le n-l\le C$. Then
\begin{equation}
|v(r)|\le C'  n^{-1/2} \left\{\begin{aligned}
&\bigl(\frac{n^2}{r}\bigr)^s \qquad && r\ge Cn^2,\\
&1                      && C^{-1}n^2\le r\le Cn^2,\\
&\bigl(\frac{\langle r\rangle }{n^2}\bigr)^s &&r\le C^{-1}n^2.
\end{aligned}\right.
\label{eqn-A.43}
\end{equation}
\end{proposition}

We would need the following

\begin{corollary}\label{cor-A.11}
As $r\le 2r^*$ 
\begin{equation}
|v(r)|\le C' r^{1/4} n^{-3/2}.
\label{eqn-A.44}
\end{equation}
\end{corollary}

\begin{proof}
It follows from Propositions~\ref{prop-A.4}--\ref{prop-A.10}.
\end{proof}

\chapter{Properties of eigenvalues and eigenfunctions of Coulomb-Schr\"odinger operator  in relativistic settings}
\label{sect-B}

Consider first the negative spectra of operators 
\begin{gather}
H^0_\beta\coloneqq H_{\beta, V^0}=T_\beta -\frac{1}{r}
\label{eqn-B.1}\\
\intertext{in $\sL^2(\bR^3,\bC)$  where $T_\beta$ is defined by (\ref{eqn-1.16}) and $r=|x|$. We assume that }
 0< \beta \le \frac{2}{\pi}
 \label{eqn-B.2}
 \end{gather}
 and compare it with non-relativistic operator operator $H^0\coloneqq - \Delta -\frac{1}{r}$, the eigenvalues  of which 
 $\lambda^0_n=-\frac{1}{4n^2}$  (of multiplicity $n^2$) is well-known.

Note that 

\begin{claim} \label{eqn-B.3}
Negative spectrum of  operator $H^0_\beta $ consists of eigenvalues $\mu_{n,l}^0\coloneqq \mu_{n,l}^0(\beta)$ ($l=0,1,\ldots, n-1$) of multiplicities $(2l+1)$ (some of those could coincide) which are 
eigenvalues of operator
\begin{gather}
K_l(\beta) = \sqrt{\beta^{-2}\bigl(-\partial_r^2 +\frac{l(l+1)}{r^2}\bigr)+\frac{1}{4}\beta^{-4}} -\frac{1}{2}\beta^{-2}-\frac{1}{r}
 \label{eqn-B.4}\\
 \intertext{in $\sL^2(\bR^+,\bC)$ and  eigenfunctions of $H^0_\beta $ are $R_{n,l}(r;\beta) Y_{l,m}(\phi, \theta)$ where}
R_{n,l}(r;\beta)= r^{-1} w_{n,l}(r;\beta),
 \label{eqn-B.5}
 \end{gather}
$w_{n,l}(r;\beta)$ are corresponding orthonormal eigenfunctions of $K_l(\beta)$  and $n=l+1, l+2,\ldots$.
\end{claim}

Recall that $Y_{l,m}(\phi, \theta)$ are spherical harmonics. 

\begin{proposition}\label{prop-B.1}
Under assumption \textup{(\ref{eqn-B.2})} the following hold:
\begin{gather}
\mu_{n,l}(\beta) < \mu_{n,l}(\beta') \le \lambda_{n}^0 =-\frac{1}{4n^2} \qquad \text{for\ \ } \beta' <\beta , 
\label{eqn-B.6}\\
\mu^0_{n,l}(\beta) < \mu ^0_{n+1,l+1}(\beta).
\label{eqn-B.7}
\end{gather}
\end{proposition}

\begin{proof}
Observe that 
\begin{multline}
T_{l}(\beta)\coloneqq \sqrt{\beta^{-2}\Lambda +\frac{1}{4}\beta^{-4}} -\frac{1}{2}\beta^{-2}\\[3pt]
=\Bigl(\sqrt{4\beta^{2}\Lambda +1} +1\Bigr)^{-\frac{1}{2}} 2\Lambda
\Bigl(\sqrt{4\beta^{2}\Lambda +1} +1\Bigr)^{-\frac{1}{2}}
 \label{eqn-B.8}
\end{multline}
with  $\Lambda\coloneqq \Lambda_{n,l} = -\partial_r^2-\frac{l(l+1)}{r^2}$. Then $T_l (\beta') < T_l(\beta)\le \Lambda$
(\ref{eqn-B.6})  follows from the variational principle, since $\lambda_n^0$ are eigenvalues of operator
\begin{equation}
K_l(0)\coloneqq  \Lambda -\frac{1}{r}.
 \label{eqn-B.9}
\end{equation}

Since $T_{l+1} > T_l$, and $(n-l)$ numbers eigenvalues, starting from the first, we have (\ref{eqn-B.7}).
\end{proof}

\begin{proposition}\label{prop-B.2}
Let $l\le \epsilon n$\,\footnote{\label{foot-15} We cover case $l\ge \epsilon n$ separately.}. Then the following estimate holds:
\begin{equation}
\mu_{n,l}^0 \ge \lambda_{n}^0 - Cn^{-3}.
 \label{eqn-B.10}
\end{equation}
\end{proposition}

\begin{proof}
Consider eigenvalue counting functions $N_{l,\beta} (\lambda)$ and $\N_{l} (\lambda)$ for $1$-dimensional operators $K_l(\beta)$ and $K_l(0)$ respectively, $\lambda<0$. Under assumption
\begin{equation}
l\le \epsilon |\lambda| ^{-1/2} 
 \label{eqn-B.11}
\end{equation}
the standard semiclassical methods show that
\begin{equation}
\N_{l,\beta} (\lambda) = \cN^\W _{l,\beta} (\lambda)+O(1),\qquad
\N_{l} (\lambda) = \cN^\W _{l} (\lambda)+O(1),
\label{eqn-B.12}
\end{equation}
where $\cN^\W_{l,\beta} (\lambda)$ and $\cN^\W_{l} (\lambda)$ are corresponding Weyl expressions.

Indeed, under assumption (\ref{eqn-B.11}) these operators are microhyperbolic (with effective semiclassical parameter $h\asymp r^{-1/2}$) and 
the remainder estimate is $C\int_{r_*}^{r^*} r^{-1}\,dr= O(\log (r^*/r_*))$  with $r^*\asymp |\lambda|^{-1}$,
$r_*\asymp (l+1)^{-2}$ and it is $O(1)$ as $l\asymp |\lambda|^{-1/2}$. 

Further,  considering propagation of singularities in the direction of increasing $r$, one can update the remainder estimate to 
\begin{equation*}
C\int_{r_*}^{r^*} \bigl( r^{-1+\delta}r^{*\,-\delta} + r^{-s}\bigr) \,dr=O(1)
\end{equation*}
with $\delta>0$.

 Finally, one can see easily that the contribution of the zone $\{r\colon r\le C_0\}$ to both $\N_{l,\beta}(\lambda)$ and $\cN^\W_{l,\beta}(\lambda)$ does not exceed $C$.

On the other hand, one can see easily that $ \cN^\W _{l,\beta} (\lambda)= \cN^\W _{l} (\lambda)+O(1)$ and since 
$\cN^\W _{l} (\lambda)\asymp |\lambda|^{-1/2}$ for $r_*\le \epsilon r^*$ we arrive to (\ref{eqn-B.10}) in this case. Here we calculate  $r^*$ exactly as in Appendix~\ref{sect-A} for non-relativistic operator $K_l(0)$. One can see easily that assumption (\ref{eqn-B.11}) is equivalent to $l\le \epsilon n$ (with different but still small constant $\epsilon>0$). We leave easy details to the reader.
\end{proof}

Recall that $w_{n,l}\coloneqq w_{n,l}( r;\beta)$ are eigenfunctions of $K_l$, $\|w_{n,l}\|=1$ where $\|.\|$ and $\blangle.,.\brangle$ are norm and an inner product in  $\sL^2 (\bR^+,\bC)$. 

\begin{proposition}\label{prop-B.3}
Under assumption \textup{(\ref{eqn-B.2})} the following estimates hold with arbitrarily large exponent $s$:
\begin{align}
&|w_{n,l}(r)|\le Cr^{-s} &&\text{for\ \ } r\ge C_0n^2
\label{eqn-B.13}\\
\shortintertext{and}
&|w_{n,l}(r)|\le C\3 w_{n,l}\3_{l;s}\coloneqq C \sum _{k\ge 0} t_k^{-s} \| \phi_{t_k} w_{n,l}\| &&\text{for\ \ } r\le C^{-1}_0l^2\,,
\label{eqn-B.14}\
\end{align}
where $\phi_t(r)= \phi(r/t)$ and $\phi \in \sC_0^\infty ([-1,1])$, $\phi(r)=1$ on $(-\frac{1}{2},\frac{1}{2})$, $t_k=2^kl$.
\end{proposition}

\begin{proof}
\begin{enumerate}[label=(\alph*), wide, labelindent=0pt]
\item\label{pf-B.3-a}
Proof is standard: observe that for $w\coloneqq w_{n,l}$
\begin{multline}
|\blangle (K_{l}(\beta) -\mu_{n,l}^0)\psi _t w ,\, \psi_t w\brangle|\le C| \blangle [[K_{l}(\beta), \psi_t]v,\psi_t]w, w\brangle |\\[3pt]
\le C_0 t^{-2} \|\psi _{t'}v\|^2 + Ct^{-2s}\|w\|^2\,,
\label{eqn-B.15}
\end{multline}
where  $1-\psi \in \sC_0^\infty ([-1,1])$, $\psi(r)=0$ on $(-1-\varepsilon,1+\varepsilon)$, $t'= (1+2\varepsilon)^{-1}t$.

On the other hand, for $t \ge C_0  n^{-2}$  
\begin{equation}
\blangle (K_{l}(\beta) -\mu_{n,l}^0)\psi _t w ,\, \psi_t w\brangle \ge C n^{-2}\|\psi _t w \|^2 - Cn{-2s}\|w\|^2\,.
\label{eqn-B.16}
\end{equation}
Combining (\ref{eqn-B.15}) and  (\ref{eqn-B.16}) and iterating we arrive to $\|\psi _t w \|\le Ct^{-s}$. Then
\begin{multline*}
\| (\beta ^{-2}\Lambda + \beta^{-4})^{1/2}\psi _t w\| \le \beta^{-2}\|\phi_t w\|  + Ct^{-2s} \implies\\[3pt]
\blangle (\beta ^{-2}\Lambda + \beta^{-4}) \psi_t w,\, \psi_t w\brangle \le \beta^{-2}\|\psi_t w\|^2 + 2C\beta^{-2} l^{-2s}  \implies\\[3pt]
\blangle \Lambda  \psi_t w,\, \psi_t w\brangle\le Ct^{-2s}\,.
\end{multline*}
Then\begin{equation}
\|\partial  \psi_t  w\| + l\|r^{-1}\psi_t w\|  \le Ct^{-s}\,,
\label{eqn-B.17}
\end{equation}
which implies (\ref{eqn-B.13}).

\item\label{pf-B.3-b}
The proof of (\ref{eqn-B.3}) follows the same scheme albeit now 
$\psi \in \sC_0^\infty ([-1,1])$, $\psi(r)=0$ on $(-1+\varepsilon,1-\varepsilon)$, $t'= (1-2\varepsilon)^{-1}t$ and we select $t= C_0^{-1}l^2$. Then
(\ref{eqn-B.15}) is replaced by 
\begin{equation}
|\blangle (K_{l}(\beta) -\mu_{n,l}^0)\psi _t w ,\, \psi_t w\brangle| \le C_0 t^{-2} \|\psi _{t'}v\|^2 + C\3 w_{n,l}\3_{l;s} \,.
\label{eqn-B.18}
\end{equation}
Then (\ref{eqn-B.16}) is replaced by
\begin{equation}
\blangle (K_{l}(\beta) -\mu_{n,l}^0)\psi _t w ,\, \psi_t w\brangle \ge C l^2t^{-1}\|\psi _t w \|^2 - C\3 w_{n,l}\3_{l;s}\,,
\label{eqn-B.19}
\end{equation}
and combining (\ref{eqn-B.18}) and (\ref{eqn-B.19}) we arrive to $\|\psi _t w \|\le C\3 w_{n,l}\3_{l;s}$. Next two inequalities are preserved, but $t^{-2s}$ is replaced by $\3 w_{n,l}\3_{l;s}$ there.
\end{enumerate}
\vskip-\baselineskip
\end{proof}

\begin{proposition}\label{prop-B.4}
Under assumption \textup{(\ref{eqn-B.2})} the following estimate hold swith arbitrarily large exponent $s$:
\begin{gather}
\| P_{l,s} w_{n,l} \| \le C \3 w_{n,l}\3_{l;s}\,,
\label{eqn-B.20}\\
\shortintertext{with}
P_{l,s}\coloneqq   - \partial _r^2 + V_s(r,\lambda_n, l,\mu ^0_{n,l})+ W_s (r,\lambda_n, l,\mu ^0_{n,l}) \partial_r \,,
\label{eqn-B.21}\\
\shortintertext{and}
V_s (r,\lambda_n,l,\mu )\coloneqq \frac{l(l+1)}{r^2} -\frac{1}{r}-\mu  + \sum_{j,k \colon 2\le j+k\le s} V_{j,k}(r/r_*, \beta) r^{-j}\mu^k\,,
\label{eqn-B.22}\\
\shortintertext{and}
W_s (r,\lambda_n,l,\mu )\coloneqq  \sum_{j,k \colon 2\le j+k\le s} W_{j,k}(r/r_*,\beta) r^{-j}\mu^k\
\label{eqn-B.23}
\end{gather}
with coefficients which could be decomposed asymptotically in positive powers of $\beta$ and non-positive of $r/r_*$.
\end{proposition}

\begin{proof}
Due to Proposition~\ref{prop-B.3} we need to consider only $r\le C_0n^2$. Then semiclassical microlocal methods imply that 
\begin{equation}
\| (1-\varphi_\tau( D_r)) \psi_t(r) w_{n,l}\|\le C\3 w_{n,l}\3_{s,l}
\label{eqn-B.24}
\end{equation}
provided $\psi\in \sC_0^\infty ([\frac{1}{2},1)$, $\varphi \in \sC_0^\infty ([-1,1])$, $\varphi =1$ on $(-\frac{1}{2},\frac{1}{2})$ and $\tau= t^{-1/2+\delta}$ with arbitrarily small exponent $\delta>0$, $t\le C_0n^2$.

Then $T_l w_{n,l} $ could be approximated in this zone by the (asymptotic) $\sum_{k\ge 0} c_k \beta^k \Lambda^{k+1}w_{n,l}$ with $c_0=1$,
and it equals to 
\begin{multline*}
\sum_{k\ge 0} c_k  \beta^k  \Lambda ^{k}(\frac{1}{r}+\mu_{n,l}) w_{n,l}\\=
\sum_{k\ge 0} c_k   \beta^k   (\frac{1}{r}+\mu_{n,l})\Lambda^{k}w_{n,l} +
\sum_{k\ge 0} c_k   \beta^k[\Lambda^{k},  (\frac{1}{r}+\mu_{n,l})]w_{n,l}
\end{multline*}
Giving weights $-\frac{1}{2}$ to $\partial _r$,  $-1$ to $\Lambda$  and $r^{-1}$, and $0$ to $l(l+1)/r$, we see that $[\Lambda,r^{-1}]$ 
has weight of $-\frac{5}{2}$ and both $[\Lambda,\partial_r]$ and  $[r^{-1},\partial_r]$  have weight of $-2$.

In this process we replace $-\partial_r^2$ by $\Lambda -\frac{l(l+1)}{r^2}$, and when $\Lambda$ reaches $w_{n,l}$ we replace it by $-(\frac{1}{r}+\mu_{n,l})$.

As a result  we have 
\begin{gather}
\bigl(K_{l}  -\mu_{n,l}^0\bigr)w_{n,l}  \equiv \bigl(I+Q     \bigr)P_{l,s} 
\label{eqn-B.25}\\
\shortintertext{with}
 Q=\beta^2 \sum_{j,k\colon 1\le j+k\le s} d_{j,k}(\beta) r^{-j}\mu^k + \beta^2 \sum_{j,k\colon 1\le j+k\le s} d'_{j,k}(\beta) r^{-j}\mu^k\partial_r\,.
 \label{eqn-B.26}
\end{gather}

It implies the statement of the proposition.
\end{proof}

\begin{proposition}\label{prop-B.5}
Under assumption \textup{(\ref{eqn-B.2})} the following estimate hold swith arbitrarily large exponent $s$:
\begin{gather}
|v_{n,l}(x)|\le C r^{1/4} n^{-3/2} \qquad\text{for\ \ }  r_*\le r \le r^*
 \label{eqn-B.27}
 \end{gather}
with $r_*=\epsilon_0 (l+1)^2$ and  $r^*=C_0 n^2$.
\end{proposition}

\begin{proof}
We can get rid of the last term in expression (\ref{eqn-B.21}) by the substitution $w_{n,l}= e^{\Phi}\tilde{w}_{n,l}$ with bounded $\Phi$, because this term is $O(r^{-2}+\mu^2)=O(r^{-2})$ for $r\le r^*$  provided $|\mu |\le  c_0n^{-2}$, which is guaranteed by (\ref{eqn-B.10}) (and it is the only place where we need (\ref{eqn-B.2}). Then we can apply arguments of Appendix~\ref{sect-A} and derive (\ref{eqn-B.27}), uniform by $\beta \le c_0$. 

Namely, let us repeat with the minor modifications arguments of Appendix~\ref{sect-A}. Then

\begin{enumerate}[label=(\alph*), wide, labelindent=0pt]
\item\label{pf-B.5-a}
Note that we do not claim that $w_{n,l}(r)$ has exactly $(n- l -1)$ zeroes as $v_{n,l}(r)$ has\footnote{\label{foot-16} See Proposition~\ref{prop-A-1}.} .

\item\label{pf-B.5-b}
On the other hand, as $r\ge C$  (large enough constant) the distance between consecutive zeroes is at least $\epsilon ' r^{1/2}$  while for 
\begin{equation*}
C(l+1)^2 \le r \le (1-\epsilon)^{-1} \mu_{n,l}^{0\,-1}
\end{equation*}
this distance does not exceed $Cr^{1/2}$. More precisely, (\ref{eqn-A.10}) and (\ref{eqn-A.11}) hold\footnote{\label{foot-17} The latter, only as $l\ge C$.} and also Proposition~\ref{prop-A.3}, Statements~\ref{prop-A.3-i} and \ref{prop-A.3-ii}\footref{foot-17}  hold. 

\item\label{pf-B.5-c}
Also Proposition~\ref{prop-A.4} and Remark~\ref{rem-A.5} hold\footnote{\label{foot-18} Under extra assumption $r_k\ge C$.}.

\item\label{pf-2.5-d}
Also Propositions~\ref{prop-A.6},~\ref{prop-A.7},~\ref{prop-A.8} and hold~\ref{prop-A.9} \footnote{\label{foot-19} With all statements, referring to $r_*$ only as $r\ge C$.}.

\item\label{pf-2.5-e}
Also Propositions~\ref{prop-A.10} holds.
\end{enumerate}

\begin{note}
We need to remember that equation is fulfilled modulo $O(r^{-s})$ but it would matter only if $|w|+|w'|=O(r^{-s'})$ but then in the process we can prove that it is not the case.
\end{note}
\vskip-\baselineskip
\end{proof}

 \begin{proposition}\label{propo-B.6}
Under assumption \textup{(\ref{eqn-B.2})} the following estimate hold swith arbitrarily large exponent $s$:
\begin{gather}
|\mu_{n,l}^0(\beta) - \mu_{n,l}^0(\beta)|\le  C|\beta^2-\beta'^2| n^{-3}(l+1)^{-1};
 \label{eqn-B.28}\\
 \intertext{in particular}
\mu_{n,l}^0(\beta) \le \bar{\lambda}_{n} -  C\beta^2 n^{-3}(l+1)^{-1};
 \label{eqn-B.29}\\
\intertext{recall that  $K_l (\beta)<K_l(\beta')<K_l(0)$ for $0<\beta' <\beta$ and therefore}
 \mu_{n,l}^0(\beta) <\mu^0_{n,l} (\beta')< \bar{\lambda}_n
  \label{eqn-B.30}
  \end{gather}
 \end{proposition}
 
 \begin{proof}
 \begin{gather}
\mu^0 _{n,l}(\beta) -\lambda^0_n = 
\int _{\beta'}^{\beta} \blangle (\partial_{\gamma^2} K_{n,l}(\gamma) w_{n,l},\, w_{n,l}\brangle \,d\gamma^2.
 \label{eqn-B.31}
\end{gather}
and we use (\ref{eqn-B.27}).
\end{proof}

\begin{corollary}\label{cor-B.7}
Consider $\mu_{n,l}^0$ as eigenvalues of the corresponding $3D$-operator. Then only two cases are possible:

\begin{enumerate}[label=(\roman*), wide, labelindent=0pt]
\item\label{cor-B.7-i}

If $\beta \le \epsilon_1$ with sufficiently small constant $\epsilon_1$, then all these eigenvalues $\mu^0_{n,l}$ are in the clusters $\cC^0_n$ of the width $\le C_0\beta^2 n^{-3}$, separated by gaps of the width $\asymp n^{-3}$. 

Each cluster contains exactly $n^2$ eigenvalues.

\item\label{cor-B.7-ii}
Otherwise, if $\beta> \epsilon_1$,  we can break $\mu^0_{n,l}$  into clusters $\cC^0_n$ of the width $\asymp n^{-3}$, separated by gaps of the width $\asymp n^{-3}$. In this case $\mu_{n,l}^0\in \cC_n$ for $l\ge C_0$  but for $l < C_0$ it could belong to $\cC^0_m$ with $n-C_0 \le m \le n$.

Each cluster $\cC^0_n$  contains  $\nu_n\in [n^2-C_1,\ n^2+C_1]$ eigenvalues.
\end{enumerate}
\end{corollary}

\begin{remark}\label{rem-B.8}
\begin{enumerate}[label=(\roman*), wide, labelindent=0pt]

\item\label{rem-B.8-i}
We do not know what is the case as  (\ref{eqn-B.2}) ever occurs under assumption (\ref{eqn-B.2}) and it would be intersting to learn. However  it does not matter for our purposes: we simply change numeration of $\mu^0_{n,l}$ so that $\mu^0_{n,l}$ with
$n=l+1,\ldots, l+\nu_n$ belongs to $\cC^0_n$.

\item\label{rem-B.8-ii}
Probably, in addition to (\ref{eqn-B.28}) the opposite inequality also holds  
\begin{gather}
|\mu_{n,l}^0(\beta) - \mu_{n,l}^0(\beta)|\ge  \epsilon_0 |\beta^2-\beta'^2| n^{-3}(l+1)^{-1}.
 \label{eqn-B.32}
 \end{gather}

\item\label{rem-B.8-iii}
It would be interesting to derive an estimate for $\|w_{n,l}(\beta)-w_{n,l}(\beta')\|$  and a poinwise estimate for $|w_{n,l}(\beta)-w_{n,l}(\beta')|$.
\end{enumerate}
\end{remark}

\end{appendices}

\end{document}